\documentclass[11pt]{amsart}

\usepackage{tikz-cd}
\usepackage{adjustbox}
\usepackage{cite}
\usepackage{hyperref}
\usepackage{amsaddr}

\allowdisplaybreaks

\usepackage{amsmath,amsthm,amsfonts,amssymb,latexsym,verbatim,bbm}
\usepackage{subcaption}
\usepackage[shortlabels]{enumitem}
\usepackage{afterpage}
\usepackage{mathtools}
\usepackage{mathrsfs}
\usepackage{hyperref}
\usepackage[capitalize]{cleveref}
\numberwithin{equation}{section}
\theoremstyle{definition}
\newtheorem{theorem}{Theorem}[section]

\newtheorem{lemma}[theorem]{Lemma}
\newtheorem{proposition}[theorem]{Proposition}

\newtheorem{example}[theorem]{Example}
\newtheorem{definition}[theorem]{Definition}

\newtheorem{remark}[theorem]{Remark}
\newtheorem{corollary}[theorem]{Corollary}

\newcommand{\iso}{\cong}
\newcommand{\arr}{\rightarrow}

\newcommand{\R}{\mathbb{R}}
\newcommand{\C}{\mathbb{C}}
\newcommand{\Z}{\mathbb{Z}}
\newcommand{\N}{\mathbb{N}}

\newcommand{\Id}{\mathbbm{1}}

\newcommand{\mcA}{\mathcal{A}}
\newcommand{\mcB}{\mathcal{B}}
\newcommand{\mcC}{\mathcal{C}}

\newcommand{\mcG}{\mathcal{G}}

\newcommand{\mcI}{\mathcal{I}}

\newcommand{\mcL}{\mathcal{L}}

\newcommand{\mcO}{\mathcal{O}}

\newcommand{\mcS}{\mathcal{S}}

\newcommand{\mcU}{\mathcal{U}}
\newcommand{\mcV}{\mathcal{V}}

\newcommand{\mbC}{\mathbb{C}}

\newcommand{\mbR}{\mathbb{R}}

\newcommand{\msA}{\mathscr{A}}
\newcommand{\msB}{\mathscr{B}}
\newcommand{\msC}{\mathscr{C}}

\newcommand{\msF}{\mathscr{F}}
\newcommand{\msG}{\mathscr{G}}

\newcommand{\msR}{\mathscr{R}}
\newcommand{\msS}{\mathscr{S}}

\newcommand{\msU}{\mathscr{U}}

\newcommand{\sgn}{\textnormal{sgn}}

\newcommand{\tr}{\text{tr}}

\usepackage[top=1.2in,bottom=1.2in,left=1.4in,right=1.4in]{geometry}
\usepackage{braket}

\begin{document}

\title[Rounding near-optimal quantum strategies]{Rounding near-optimal quantum
strategies for nonlocal games to strategies using a maximally entangled state}

\author[C.~Paddock]{Connor Paddock}
\address{Department of Mathematics and Statistics, University of Ottawa, Canada}
\email{cpaulpad@uottawa.ca}

\begin{abstract}
We establish approximate rigidity results for several well-known families of nonlocal
games. In particular, we show that near-perfect quantum strategies for boolean constraint
system (BCS) games are approximate representations of the corresponding BCS algebra. Likewise,
for the class of XOR nonlocal games, we show that near-optimal quantum strategies are
approximate representations of the corresponding $*$-algebra associated with optimal
quantum values for the game. In both cases, the approximate representations are with
respect to the normalized Hilbert-Schmidt norm and independent of the Hilbert space or
quantum state employed in the strategy.

We also show that approximate representation of the BCS (resp.~XOR-algebra) yields
measurement operators for near-perfect (resp.~near-optimal) quantum strategies where the
players employ a maximally entangled state in the game. As a corollary, every near-perfect
BCS (resp.~near-optimal XOR) quantum strategy is close to a near-perfect
(resp.~near-optimal) quantum strategy using a maximally entangled state. Lastly, we
establish that every synchronous algebra is $*$-isomorphic to a certain BCS algebra called
the SynchBCS algebra. This allows us to apply our BCS rigidity results to the class of
synchronous games as well.
\end{abstract}

\maketitle

%%%%%%%%%%%%%%%%%%%%%%%%%%%%%%%%%%%%%%%%%%%%%%%%%%

\section{Introduction}\label{sec:intro}

A two-player nonlocal game is a scenario involving two players, commonly referred to as
Alice and Bob, and a referee. In the game, the referee sends each player a question, and
each player replies with an answer. The players are unable to communicate once the game
begins. However, they may share a bipartite quantum state and perform measurements on the
state as part of their strategy. The players win if their answers satisfy the \emph{rule
predicate}, otherwise they lose. The rule predicate is known to the players beforehand
allowing them to predetermine their strategy.

It is well-known that there are nonlocal games where by using entanglement the players can
win with a higher probability than if they had only classical resources. There are even
examples of nonlocal games where the players can win \emph{perfectly} (with probability
one) using an entangled strategy, while any classical strategy for the game has a nonzero
losing probability. However, determining if a nonlocal game admits some quantum advantage
is not easy. Not only can it be hard to find the optimal winning probability amongst
classical strategies, a quantity known as the \emph{classical value}, it was recently
established that deciding if the optimal winning probability for a nonlocal game amongst
all entangled strategies, a quantity known as the \emph{entangled value}, is $1$ or
greater than $1/2$ is equivalent to the halting problem \cite{JNVWY22}.

Despite the computational hardness surrounding and entangled value; for several classes of
nonlocal games, the existence of optimal and or perfect quantum strategies can be
characterized in purely algebraic terms. This notion is often referred to as the
``rigidity property'' of the optimal strategies for a nonlocal game. Abstractly, the
rigidity relations amongst the observables in an optimal strategy can be viewed as
generators and relations of a finitely presented $*$-algebra associated with the nonlocal
game. By construction, the finite-dimensional representations of these nonlocal game
algebras yield quantum strategies that obtain the quantum (or commuting operator) value
for such games.

The first instance of this correspondence is in the context of XOR nonlocal games through
the algebraic characterization of specific quantum correlations in two-output Bell
scenarios and is attributed to Tsirelson \cite{Tsi85}. Tsirelson's result implies that
the optimal quantum strategies for XOR games are representations of a certain finitely
presented $*$-algebra with a finite-dimensional tracial state \cite{Tsi87,Slof11}. There
are now several families of nonlocal games for which a correspondence between optimal
quantum strategies and representations of the nonlocal game algebra has been established.
This includes the class of $\Z_2$-linear constraint systems (LCS) games \cite{CM14,
CLS17}, the class of synchronous games \cite{PSSTW16}, boolean constraint system (BCS)
games \cite{CM14,Ji13}, and the general class of imitation games \cite{Lup20}.

A particularly useful application of this correspondence is in providing lower bounds on
the amount of entanglement required to achieve the quantum value of the game, for
example, \cite{Slof11,Slof18}. Another application of this correspondence is enabling an
observer to verify information about the quantum measurements and states employed by the
players when they achieve the optimal winning probability in a game. This concept is more
formally known as self-testing and is an important ingredient in device-independent
cryptography, see for example \cite{MY03,Wu16,BSCA18a,BSCA18b,Kan19}.

With these applications in mind, it is natural to wonder how the correspondence between
optimal quantum strategies and representations is affected by noise. More precisely, we
say a quantum strategy is \emph{$\epsilon$-optimal} if the probability of winning is at
most $\epsilon$-away from the entangled value. In the case where the entangled value is
one, we say that a quantum strategy is \emph{$\epsilon$-perfect} if it wins with
probability at least $1-\epsilon$. In this work, we focus on the case where $\epsilon$ is
significantly less than the smallest joint question probability\footnote{When this is not
the case, the noise from the state is indistinguishable from a strategy employing losing
answers and other issues arise which we do not explore here.}. In this regime, we will see
that $\epsilon$-optimal strategies correspond to \emph{approximate representations}.
Informally, an approximate representation or $\epsilon$-representation of a
finitely presented $*$-algebra is a function from the generators to matrices where the
defining relations hold approximately. The parameter $\epsilon>0$ measures how far,
according to some metric, the relations are from being satisfied.

There are already several previous results about \emph{approximate rigidity} in the
literature. In \cite{Slof11}, Slofstra showed that the correspondence between optimal
quantum strategies and representations of the XOR-algebra is robust, in the sense that
$\epsilon$-optimal strategies are $O(\epsilon^{1/8}d^{2/3})$-representations of the
XOR-algebra, where $d$ is the dimension of the local strategy Hilbert space $H_A$ (or
equivalently $H_B$) supporting the quantum strategy. In the case of $\Z_2$-linear
constraint systems (LCS) nonlocal games, Slofstra and Vidick established that
$\epsilon$-perfect quantum strategies correspond to unitary
$O(\epsilon^{1/4})$-representations of the associated solution group \cite{SV18}. Unlike
in the XOR game case, for LCS games the quality of the approximate representation does not
depend on the Hilbert space or the state in the quantum strategy. This independence is a
much-desired property in the context of device independence.

Our main result is that the correspondence between perfect (or optimal in the XOR case)
quantum strategies and representations is robust and Hilbert space independent for the
class of boolean constraint system (BCS), synchronous, and XOR nonlocal games.

\begin{theorem}\label{thm:main}\
\begin{enumerate}
\item Each $\epsilon$-perfect strategy for a BCS nonlocal game corresponds to an
$O(\epsilon^{1/4})$-representation of the BCS algebra. \item Each $\epsilon$-perfect
strategy for a synchronous nonlocal game corresponds to an
$O(\epsilon^{1/8})$-representation of the synchronous algebra. \item Each
$\epsilon$-optimal strategy for an XOR nonlocal game corresponds to an
$O(\epsilon^{1/8})$-representation of the XOR algebra.
\end{enumerate}
\end{theorem}

The precise definitions of a $\epsilon$-perfect and optimal strategies are given in
\cref{sec:nonlocal_games}, along with the definitions of the BCS algebra,
and the XOR algebra, while approximate representation are fomally defined in
\cref{sec:approx_reps} \cref{defn:eps_rep}. The more precise statements of \cref{thm:main}
are stated in \cref{prop:res1}, \cref{prop:res2}, and \cref{prop:res3}. All of the
approximate representations are measured with respect to the little Frobenius norm
$\|\cdot\|_f$. In particular,\cref{thm:main}(1) can be seen as a generalization of the
result in \cite{SV18} to the class of more general BCS nonlocal games. \cref{thm:main}(2)
provides an alternative proof of the result in \cite{Vid22} in the case of games. While
\cref{thm:main}(3) can be seen as an improvement of the result in \cite{Slof11}, as it
removed the Hilbert space dependence in the approximate representation.

The proof of the first theorem consists of two parts. First, we establish that every
near-perfect (resp. near-optimal) strategy is an approximate representation of the BCS
(resp. XOR) algebra with respect to a particular state-dependent semi-norm. This
state-dependent semi-norm is determined by the quantum state employed as part of the
quantum strategy used by the players. The second step involves showing that each
state-dependent approximate representation can be ``rounded'' to an approximate
representation in the little Frobenius norm. This removes the state/dimension dependence
in the approximate representation. The rounding of the state-dependent approximate
representations from near-optimal quantum strategies to an approximate representation in
the little Frobenius norm is achieved through \cref{lem:rounding}, which builds on the
techniques developed in \cite{SV18} in the group setting.

Although our results do not depend on the dimension of the approximate representation, the
approximate representations do depend on the properties of the game algebra. In
particular, this means that approximate representations may depend (exponentially in some
cases) on the size of the question and answer sets from the nonlocal game. This means that
although the techniques apply to fixed games, they do not apply to families of games with
these parameters. We leave the problem of tightly analyzing this dependence for future
work.

Another contribution of this work is furthering the connection between synchronous games
and boolean constraint system (BCS) games. In \cref{prop:iso}, we establish that the
synchronous algebra is isomorphic to the BCS algebra of a certain BCS nonlocal game we
call the SynchBCS game. This isomorphism allows us to extend several of our results to the
class of synchronous algebras. Connections between synchronous and BCS games have been
previously noted, but have focussed exclusively on the case where the constraints are all
linear, for example \cite{KPS18,Gol21,Fri20}.

Our second main result is a converse to \cref{thm:main}. In particular, we show that
approximate representations for certain game algebras are close to near-perfect
(near-optimal in the XOR case) quantum strategies employing a maximally-entangled state.

\begin{theorem}\label{thm:second}\
\begin{enumerate}
\item Each $\epsilon$-representation of the BCS algebra is close to an
$O(\epsilon^2)$-perfect strategy for the corresponding BCS nonlocal game employing a
maximally entangled state. \item Each $\epsilon$-representation of the synchronous algebra
is close to an $O(\epsilon^2)$-perfect strategy for the synchronous nonlocal game where
the players employ a maximally entangled state. \item Each $\epsilon$-representation of
the XOR algebra is close to an $O(\epsilon^2)$-optimal strategy $\mcS$ for the XOR
nonlocal game where the players employ a maximally entangled state.
\end{enumerate}
\end{theorem}

By \emph{close} we mean that each measurement operator in the strategy is at most
$O(\epsilon)$-away from the generators in the approximate representation. Because all of
the above algebras are quotients of group algebras, we assume that the operator norm of
each generator in the approximate representation is at most one. In this case the
$O(\epsilon)$ is entirely independent of the Hilbert space and depends only on the
presentation of the nonlocal game algebra. The formal statements for \cref{thm:second} are
found in \cref{prop:BCS_alg}, \cref{prop:synch_alg}, and \cref{prop:XOR_alg}. As a
consequence of the proofs of \cref{thm:main} and \cref{thm:second} we obtain the following
corollary:

\begin{corollary}\label{cor:main}\
\begin{enumerate}
\item Each $\epsilon$-perfect quantum strategy for a BCS nonlocal game is
$O(\epsilon^{1/4})$-away from an $O(\epsilon^{1/2})$-perfect quantum strategy using a
maximally entangled state. \item Each $\epsilon$-perfect quantum strategy $\mcS$ for a
synchronous nonlocal game is $O(\epsilon^{1/8})$-away from an $O(\epsilon^{1/4})$-perfect
quantum strategy using a maximally entangled state. \item Each $\epsilon$-optimal quantum
strategy $\mcS$ for an XOR nonlocal game is $O(\epsilon^{1/8})$-away from an
$O(\epsilon^{1/4})$-optimal quantum strategy using a maximally entangled state.
\end{enumerate}\end{corollary}

Notably, \cref{cor:main} shows that you can reduce the analysis of near-perfect
(near-optimal in the XOR case) strategies with an arbitrary state to the analysis of
near-perfect (near-optimal in the XOR case) strategies with a maximally-entangled state
without amplifying the error too much.

Recently, independent but similar results to \cref{thm:main}(2) for approximate synchronous
correlations were established in \cite{Vid22}. Both results are based on techniques in
\cite{SV18} but we emphasize that this work takes a more algebraic perspective and focuses
on extending the techniques to arbitrary BCS games. One advantage of their result is that
it applies to the more general case of correlations (not just strategies). However, none
of their results apply to the case of XOR games as they are far from synchronous.

The remainder of the paper is outlined as follows: \cref{sec:prelims} covers some
mathematical preliminaries, \cref{sec:approx_reps} covers the relevant concepts and
results from Approximate Representation Theory we use in the work, including the key
rounding result (\cref{lem:rounding}), and \cref{sec:nonlocal_games} defines the nonlocal
game algebras associated with BCS, synchronous, and XOR games, while also examining the
connection between approximate representations and near-optimal strategies for these
games.

%%%%%%%%%%%%%%%%%%%%%%%%%%%%%%%%%%%%%%%%%%%%

\section{Preliminaries}\label{sec:prelims}

Let $S$ be a finite set, we let $\C^*\langle S\rangle$ denote the free (complex)
$*$-algebra generated by $S$. Let $R\subset \C^*\langle X\rangle$ be a finite
collection of elements (noncommutative $*$-polynomials) from $\C^*\langle
S\rangle$. The \textbf{finitely-presented $*$-algebra} $\msA=\C^*\langle S:
R\rangle$ is the quotient of $\C^*\langle S\rangle$ by $\langle\langle
R\rangle\rangle$, where $\langle\langle R\rangle\rangle$ is the smallest
(two-sided) $*$-ideal containing $R$. We call the pair $(S,R)$ a finite
\textbf{presentation} of the $*$-algebra $\msA$. We write $\Id$ to represent the
unit in a $*$-algebra. A priori, elements of $\msA=\mbC^*\langle S:R\rangle$ are
not bounded in representations. To address this, let $\msA_{sa}=\{a\in
\msA:a=a^*\}$ be the $*$-subalgebra of self-adjoint elements. Define the subset
of \textbf{sums-of-squares (SOS)} to be the elements $\msA_+=\{a\in
\msA:a=\sum_{k\in K} b_k^*b_k\}$. In the language of \cite{Oza13}, the subset
$\msA_+\subset \msA$ is a \textbf{$*$-positive cone} for $\msA$. The
$*$-positive cone induces a partial order on the self-adoint elements
$\msA_{sa}$, where $a\leq b$ if $a-b\in \msA_+$. The $*$-subalgebra of bounded
elements is defined as $\msA_{bdd}=\{a\in \msA: \exists\ \lambda>0\text{ such
that } a^*a\leq \lambda\Id\}$. A finitely presented $*$-algebra
$\msA=\mbC^*\langle S:R\rangle$ is said to be \textbf{archimedean} if
$\msA=\msA_{bdd}$. In this case, we say that the relations $R$ are archimedean.
For a finitely presented $*$-algebra $\C^*\langle a_1,\ldots,a_n: r_1,\ldots,
r_m\rangle$ being archimedean is equivalent to the ideal generated by $\langle
r_1,\ldots, r_m\rangle $ containing the relation $n\lambda^2\Id-\sum_{i=1}^n
a_ia_i^*$ for some $\lambda>0$, see for instance \cite{HMV25}. Additionally,
whenever $\msA=\msA_{bdd}$ we say that $\msA$ is a
\textbf{semi-pre-$C^*$-algebra}. A \textbf{representation} of $\msA$ is a
$*$-homomorphism $\psi:\msA\to B(H)$, where $B(H)$ are the bounded operators on
a Hilbert space $H$. Note that a finitely-presented $*$-algebra
$\msA=\mbC^*\langle S:R\rangle$ is archimedean if $S\subset \msA_{bdd}$. Remark
that the relation $n\lambda^2\Id-\sum_{i=1}^n a_ia_i^*$ implies that in any
$*$-representation, the image of each generator has operator norm at most
$\lambda$. Given a finitely presented $*$-algebra $\msA$ we let
$\lambda_\msA=\inf_\lambda \{n\lambda^2\Id-\sum_{i=1}^n a_ia_i^*\in
\langle r_1,\ldots, r_m\rangle \}$ denote the \textbf{bounded radius of $\msA$}.

For $A\in M_d(\C)$, $\|A\|_{op}$ denotes the \textbf{operator norm} of $A$, while
$\|A\|_F$ denotes the \textbf{Frobenius (or Hilbert-Schmidt) norm} of $A$. We write
$\|\cdot\|$ when the matrix norm is left unspecified. For a finite-dimensional Hilbert
space $H$, we denote by $\mcL(H)$ the set of linear operators from $H$ to $H$. Whenever
$H\iso \C^d$, we define $d=dim(H)$ and we have that $M_d(\C)\iso\mcL(H)$ is a Hilbert
space with the \textbf{Frobenius (Hilbert-Schmidt) inner-product} $ \langle
A,B\rangle_F:=\tr(A^*B)$, for $A,B\in \mcL(H)$. We also use the \textbf{little Frobenius
(or normalized Hilbert-Schmidt) norm}, denoted by
$\|A\|^2_f:=\tilde{\tr}(A^*A)=\tfrac{1}{d}\|A\|_F^2$, for $A\in M_d(\C)$. The
normalization in the little Frobenius norm ensures that $\|\Id\|_f=1$, in contrast to
$\|\Id\|_F=\sqrt{d}$. It's worth noting that unlike its unnormalized version, the
Frobenius norm $\|\cdot\|_F$, the little Frobenius norm $\|\cdot\|_f$ is \emph{not}
submultiplicative. Nonetheless, if $A,B,C\in M_d(\C)$ we do have the bimodule property
$\|ABC\|_f\leq \|A\|_{op}\|B\|_f\|C\|_{op}$. If $P$ is an orthogonal projection in a matrix algebra $M_d(\C)$, then $PM_d(\C)P$ is the compression of $M_d(\C)$ to the subspace supported on the image of $P$.

A (pure) \textbf{quantum state} $|\psi\rangle$ is a unit vector in a Hilbert space $H$.
Each quantum state $|\psi\rangle$ gives rise to a positive-semidefinite matrix with trace
one, called a \textbf{density matrix} (or sometimes just a state) $\rho$ through the
identification $|\psi\rangle \mapsto |\psi\rangle\langle \psi|:=\rho$. Density matrices
can also represent ensembles or mixtures of pure states. The \textbf{state-induced
semi-norm} (or \textbf{$\rho$-norm}) for a density operators $\rho\in \mcL(H)$ is given by
$\|T\|_\rho:=\|T\rho^{1/2}\|_F=$ for all $T\in \mcL(H)$. The failure of positive
definiteness in the $\rho$-norm is the result of $\rho$ having a $0$-eigenvalue. In the
case where $\rho=\Id/d$, the $\rho$-norm $\|\cdot\|_\rho$ coincides with the little
Frobenius (normalized Hilbert-Schmidt) norm, that is $\|\cdot\|_f=\|\cdot\|_{\Id/d}$.

A bipartite quantum state is a unit vector $|\psi\rangle$ in the tensor product of Hilbert
spaces $H_A\otimes H_B$. A state $|\psi\rangle\in H_A\otimes H_B$ is said to be
\textbf{maximally entangled} if its reduced density matrix $\tr_{H_A}(\rho)=\rho_B$ on
$H_B$ (or $\rho_A=\tr_{H_B}(\rho)$ on $H_A$) is $\Id/\dim(H_B)$ (or $\Id/\dim(H_B)$).
Thus, starting from a maximally entangled state the induced $\rho$-norm on $H_B$ (or
$H_A$) is the little Frobenius norm $\|\cdot\|_f$. Every bipartite vector state has a
\textbf{Schmidt decomposition} $|\psi\rangle =\sum_{i\in I} \alpha_i |u_i\rangle \otimes
|v_i\rangle$, where $\{|u_i\rangle\}_{i\in I}$ and $\{|v_i\rangle\}_{i\in I}$ are
orthonormal subsets of $H_A$ and $H_B$ respectively, and $\alpha_i>0$ for all $i\in I$.
The \textbf{support} of a bipartite vector state $|\psi\rangle$ is the image of
$H_A\otimes H_B$ under the projection $\Pi=\sum_{i\in I} |u_i\rangle\langle u_i| \otimes
|v_i\rangle\langle v_i|$. The support on $H_A$ or $H_B$ are the images under the local
projections $\Pi_A=\sum_{i\in I} |u_i\rangle\langle u_i|$ and $\Pi_B=\sum_{i\in I}
|v_i\rangle\langle v_i|$ respectively. The support Hilbert space is defined by
$Im(\Pi):=\Pi H$. Note that for the maximally entangled state, the local support
projections are $Im(\Id)=H_A$ and $Im(\Id)=H_B$. We denote the maximally entangled state
by $|\tau\rangle=|I|^{-1/2}\sum_{i\in I} |u_i\rangle \otimes |u_i\rangle$. For a self-adjoint matrix
$A^*=A$ we observe that $A\otimes \Id|\tau\rangle=\Id\otimes A^\top|\tau\rangle=\Id\otimes
\overline{A}|\tau\rangle$, where the transpose is taken with respect to the Schmidt basis
of $|\tau\rangle$, since $A^\top=(A^*)^\top=\overline{A}$. Moreover, this identification
shows that there is a correspondence between the norms $\|A\otimes
\Id|\tau\rangle\|=\|A\|_f$. More generally, if $|\psi\rangle \in H\otimes H$ is a
\textbf{purification} of $\rho\in \mcL(H)$, then $\|A\otimes
\Id|\psi\rangle\|=\|A\|_\rho$.

For positive real functions $f,g:\R_{\geq 0}\to \R_{\geq 0}$ as $x\to 0$ we write
$f(x)=O(g(x))$, if there exists constants $C,K>0$ such that for all $x\in (0,C)$ we have
that $f(x)\leq Kg(x)$. This is read as ``$f$ is big-Oh of $g$'', and means for small $x$
the behaviour of $f$ is dominated by a constant times the function $g$.

The \textbf{unitary part} of a $d\times d$ complex matrix $A$ we mean \emph{any} unitary
$U$ satisfying the \textbf{polar decomposition} equation $A=U|A|$ for $|A|=\sqrt{A^*A}$.
Every matrix has a unitary part, the simplest construction comes from the the singular
value decomposition $A=W\Sigma V$ with $U=WV$. Moreover, when $|A|$ is invertible the unitary part of $A$ is unique
and given by $ U=A|A|^{-1}=A(A^*A)^{-1/2}.$ If $A$ is self-adjoint then $\sgn(A)$ is a
(self-adjoint) unitary which satisfies $A=\sgn(A)|A|$.

%%%%%%%%%%%%%%%%%%%%%%%%%%%%%%%%%%%%%%%%%%

\section{Approximate representation theory}\label{sec:approx_reps}

In the first part of this section, we present the key definitions and concepts from
approximate representation theory. The second part of this section contains the proof of
the main technical lemma. All of the Hilbert spaces in this section are complex and
finite-dimensional unless stated otherwise. Moreover, we assume every finitely-presented
$*$-algebra is a $\C$-vector space.

\begin{definition}\label{def:lift}
Let $\msA=\C^*\langle S:R\rangle$ and $\msB=\C^*\langle T:U\rangle$ be finitely presented
$*$-algebras. If $\psi:\msA \to \msB$ is a $*$-homomorphism, then the \textbf{lift} of
$\psi$ is the unique $*$-homomorphism $\tilde{\psi}:\C^*\langle S\rangle \to \C^*\langle
T\rangle$ such that $\tilde{\psi}(r)=0$ for all $r\in R$. Equivalently, we say that
$\tilde{\psi}$ \textbf{descends} to the $*$-homomorphism $\psi$.
\end{definition}

Of particular interest is when $\msB$ is the $*$-algebra of linear operators on $H$ with
the usual antilinear involution. When $H$ is finite-dimensional, this is matrix algebra by
the standard identification $\mcL(H)\iso M_d(\C)$ where $H\iso \C^d$. In this case, the
$*$-homomorphisms (or \emph{representations}) $\msA\to \mcL(H)$ are in one-to-one
correspondence with the \emph{lifts} $\tilde{\psi}:\C^*\langle S\rangle\to \mcL(H)$. This
point of view is essential for motivating the definition of an approximate representation.

\begin{definition}\label{defn:eps_rep}
Let $\msA=\mbC^*\langle S:R\rangle$ be a finitely-presented $*$-algebra. If $H$ is a
Hilbert space, $\rho$ a state (i.e.~density operator) in $\mcL(H)$, a
\textbf{state-dependent $\epsilon$-representation} of $\msA$ or
\textbf{$(\epsilon,\rho)$-representation} is a $*$-homomorphism
	\begin{equation*}
		\phi:\mbC^*\langle S\rangle\to \mcL(H)\text{ such that }\|\phi(r)\|_\rho\leq
		\epsilon,\text{ for all $r\in R$.}
	\end{equation*}
\end{definition}

\begin{remark}
If the $\rho$-norm is non-degenerate then every $(0,\rho)$-representation descends to a
$*$-homomorphism $\msA\to \mcL(H)$. Moreover, if $T\subset R$ is a subset of relations and
$\phi$ is a non-degenerate $(\epsilon,\rho)$-representation of $\msA=\mbC^*\langle
S:R\rangle$ where $\|\phi(r)\|_\rho=0$ for all $r\in T$, then $\phi$ satisfies the
relations $T\subset R$. In particular, this means that $\phi$ descends to a representation
of $\mbC^*\langle S:T\rangle$.
\end{remark}

The universal notions identified in \cref{def:lift} illustrate why
$\epsilon$-representation $\phi$ of $\msA$ are formally $*$-homomorphisms of the free
$*$-algebra $\C^*\langle S\rangle$ such that ``$\phi(R)\approx_\epsilon 0$''. However,
they also indicate some degree of flexibility in the notion of an approximate
representation. In particular, there are multiple ways to quantify
``$\phi(R)\approx_\epsilon 0$''. As such, other notions of approximate representations
exist in the literature\footnote{The definition in \cref{defn:eps_rep} can
be viewed as a \emph{worst case} notion of approximate representation. This is contrasted
with an \emph{average case} notion, where $\epsilon$ represents the average error over all
the relations according to a measure on $R$, see for instance \cite{CVY23}.}, see for
example \cite{GH17,Thom18}. Returning to \cref{defn:eps_rep}, we highlight the important
case where the state $\rho$ is the \emph{maximally-mixed state} $\rho=\Id/d$, with
$dim(H)=d$. In this case, the semi-norm or $\rho$-norm coincides with the little Frobenius
(a.k.a.~normalized Hilbert-Schmidt) norm $\|\cdot\|_{\Id/d}=\|\cdot\|_f$. This norm has
many nice qualities and is a popular norm for studying approximate representations. We
make the following definition.

\begin{definition}
Let $H\iso \C^d$. An $(\epsilon,\rho)$-representation of $\msA=\mbC^*\langle S:R\rangle$
where $\rho=\Id/d$ is called a \textbf{state-independent} $\epsilon$-representation or
simply an \textbf{$\epsilon$-representation}. In this case, it is clear that the relations
hold approximately in the little Frobenius norm $\|\cdot\|_f$ on $\mcL(H)$.
\end{definition}

A keen reader may be aware that $*$-representations of $\C^*\langle X\rangle$ on arbitrary
(possibly infinite-dimensional) Hilbert spaces are not bounded. This means that there is
no universal bound on the operator norm of an element in an
$(\epsilon,\rho)$-representation on an arbitrary Hilbert space. That being said, in any
given $(\epsilon,\rho)$-representation every $\|\phi(x)\|_{op}$ is finite when $H$ is
finite-dimensional. Indeed, since $X$ is a finite set for any given
$(\epsilon,\rho)$-representation $\max\{\|\phi(x)\|_{op}:x\in X\}$ bounds the norm of every
$\phi(x)$. Hence, by letting $\kappa_\phi>0$ be the largest singular value among
generators of an $(\epsilon,\rho)$-representation we can explicitly track how results
about $(\epsilon,\rho)$-representations depend on this quantity. On the other hand, we
recall that if a finitely-presented $*$-algebra $\msA=\mbC^*\langle S:R\rangle$ is
\emph{archimedean} then in every $*$-representation $\psi:\mbC^*\langle S\rangle\to
\mcL(H)$ of $\msA$, the largest singular value of each $\psi(s)$ is bounded by the radius
$\vartheta_\msA>0$ of $\msA$, which depend on the presentation of $\msA$ and not on $H$.
Hence, it is not unreasonable to expect that approximate representations have the same or
similar bounded property, especially since as $\epsilon\rightarrow 0$ we expect for
$\kappa_\phi$ to coincide with $\vartheta_\msA$. To keep things simple, we can also
restrict the domain of our approximate representations when $\msA$ is archimedean.

\begin{definition}\label{def:arch_rep}
Let $\msA=\mbC^*\langle S:R\rangle$ be an archimedean finitely-presented $*$-algebra, and
let $\vartheta_\msA>0$ be the bounded radius. An $(\epsilon,\rho)$-representation
$\phi:\mbC^*\langle S\rangle\to \mcL(H)$ is a \textbf{bounded approximate representation}
if each $\phi(s)$ has singular value at most $\vartheta_\msA$, for all $s\in S$.
\end{definition}

Although the results of this section are for general approximate representations, we note
that by restricting to the class of \emph{bounded approximate representations} some result
can be strengthened, in particular when the bounded radius of $\msA$ is $1$ and we
take our approximate representations to be bounded as well.

\subsection{Stability and replacement}

One of the central questions in approximate representation theory is when are approximate
representation \emph{close} to exact representations? The answer to this question is
captured by the notion of stability for finitely presented $*$-algebras. Intuitively, the
more stable the algebra the more closely approximate representations correspond to genuine
representations.

\begin{definition}\label{def:mat_stab}
Let $g:\mbR_{\geq0} \arr \mbR_{\geq0}$ be a non-negative function. A finitely presented
$*$-algebra $\msA=\mbC^*\langle S:R\rangle$ is \textbf{$(g,C)$-stable} if for every
non-degenerate $(\epsilon,\rho)$-representation of $\msA$ given by $\phi:\C^*\langle S\rangle \arr
\mcL(H)$ with $\epsilon\leq C$, there is a $*$-homomorphism $\psi:\C^*\langle S\rangle\arr
\mcL(H)$ of $\msA$ such that
	\begin{equation*}
		\|\phi(s)-\psi(s)\|_\rho\leq g(\epsilon),
	\end{equation*}
for all $s\in S$. Alternatively, we say that $\msA$ is $g$-\textbf{stable}
if it is $(g,C)$-stable for all $\epsilon\geq 0$, and \textbf{stable} if
$g(\epsilon)=O(\epsilon)$.
\end{definition}

The stability function $g:\mbR_{\geq0} \arr \mbR_{\geq0}$ describes the behaviour of how
exact representations relate to approximate representations. The asymptotics of $g$ gives
us an idea of how much we need to perturb or shift $\phi$ to obtain a genuine
representation. We make two remarks: firstly, the notion of stability should be Hilbert
space free in the sense that it should not depend on the dimension of $H$. Secondly,
although the stability of a finitely presented algebra is sensitive to the choice of
presentation, the following result shows that for state-independent approximate
representations changing the presentation will not affect the stability asymptotically.
Despite our earlier emphasis on state-dependent approximate representations, several facts
about stability are significantly harder to establish in this regime because the
$\rho$-norm generally fails to have the bimodule property with respect to the operator
norm. So we proceed in the state-independent case and mention when a result holds for the
state-dependent case.

\begin{proposition}\label{prop:stable_pres}
Let $\msA=\C^*\langle S:R\rangle$ and $\msB=\C^*\langle T:U\rangle$ be finitely presented
$*$-algebras and $H$ a Hilbert space. If $\phi: \C^*\langle T\rangle\to \mcL(H)$ is an
$\epsilon$-representation and $\psi:\msA\to \msB$ a $*$-homomorphism, then there exists a
constant $C>0$ so that $\phi\circ \tilde{\psi}$ is an $C\epsilon$-representation of
$\msA$, where $\tilde{\psi}$ is the lift of $\psi$.
\end{proposition}

\begin{proof}
Let $\varphi:\C^*\langle S\rangle \to \msA$ and $\eta:\C^*\langle T\rangle \to \msB$ be
the quotient maps induced by the (two-sided) $*$-ideals $\langle \langle R\rangle\rangle$
and $\langle \langle U\rangle \rangle$ respectively. Furthermore, let $\tilde{\psi}$ be
the lift of the $*$-homomorphism $\psi:\msA\to \msB$. Since $\varphi\circ
\psi(r)=\eta\circ \tilde{\psi}(r)=0$ for all $r\in R$, we conclude that
$\tilde{\psi}(r)\in \langle \langle U\rangle \rangle$ for all $r\in R$. Consider a single
$r\in R$, and note that $\tilde{\psi(r)}\in \langle\langle U\rangle \rangle$ if and only
if there exists a collection $i\in I\subseteq U$, coefficients $\gamma_i\in \C$, monomials
$\{w_i, v_i\} \in \C^*\langle T\rangle$, and relations $\{u_i\} \in U$, so that
$\tilde{\psi}(r)=\sum_{i\in I} \gamma_i w_i u_i v_i$. Then, if $\phi:\C^*\langle
T\rangle\to \mcL(H)$ is a $\epsilon$-representation of $\msB$, we see that each $r$ is
bounded by
\begin{align*}
\|\phi\circ\tilde{\psi}(r)\|_f
\leq \sum_{i\in I} |\gamma_i|\|\phi(w_i)\|_{op}\|\phi(u_i)\|_f \|\phi(v_i)\|_{op}
\leq \sum_{i\in I} C_0\|\phi(u_i)\|_f
\leq |I|C_0\epsilon,
\end{align*}
where $C_0$ is a constant depending on the presentation which bound the coefficients
$\gamma_i$, and the operator norms of the monomials $\{v_i\}_{i\in I}$ and $\{u_i\}_{i\in
I}$ respectively. That is, we let
$C_0=\max_i\{|\gamma_i||\phi(w_i)\|_{op}\|\phi(v_i)\|_{op}\}$. Furthermore, if
$\kappa_\phi=\max_{t\in T}\{\|\phi(t)\|_{op}\}$, then $\|\phi(w_i)\|_{op}$ (resp.
$\|\phi(v_i)\|_{op}$) are bounded by $\kappa_\phi^l$ where $l=\max_i\{len(w_i),len(v_i)\}$
is the longest monomial in $\{v_i,w_i:i\in I\}$, which is finite but not given explicitly
from the presentation. Now, we define the constant $C_r=C_0|I|$, and the result follows by
taking the largest $C_r$ among the relations in $R$, that is $C=\max_{r\in R}\{C_r\}$.
\end{proof}

In particular, if $\C^*\langle S:R\rangle$ and $\C^* \langle T:W\rangle$ are both
presentations of a finitely presented $*$-algebra $\msA$ then $\C\langle S:R\rangle
\iso_\psi \C^*\langle T:W\rangle$ and if $\C^*\langle S:R\rangle$ is stable with
$g(\epsilon)=C\epsilon$, then there exists a constant $C'>0$ such that $\C^* \langle
T:W\rangle$ is stable with $g(\epsilon)=C'\epsilon$ and vice versa. We do note that the
constant in \cref{prop:stable_pres} does depend on $\kappa_\phi$ the operator norm of the
approximate representation, and so one should be cautious in applying this result in a
case where $\phi$ is not bounded and or $\msA$ is not archimedean. Fortunately, for our
applications this is not an issue.

Up until now the discussion of stability has been quite abstract. In reality stability for
matrices is quite a concrete notion. For example, the $*$-algebra of self-adjoint matrices
\begin{align*}
\C^*\langle X_1,\cdots,X_n: X_i^*-X_i\text{, for all $1\leq i \leq n$}\rangle
\end{align*}
is stable with $g(\epsilon)=\epsilon/2$. We can see this by remarking that for any $X_i$
with $\|X_i-X_i^*\|_f\leq \epsilon$, setting $Y_i=\frac{1}{2}\left(X_i^*+X_i\right)$ the
following two conditions holds: (i) each $Y_i$ is self-adjoint, and (ii) each $Y_i$ is
close to $X_i$, since $ \|X_i-Y_i\|_f\leq \frac{1}{2}\|X_i-X_i^*\|_f\leq \epsilon/2$. We
can repeat the construction for all $1\leq i \leq n$ to get a collection of self-adjoints
$Y_1,\ldots, Y_m$, and we see the self-adjoint relations are $\epsilon/2$-stable.

This description of stability in terms of stable relations is intentional, and it
motivates the following case. Suppose the self-adjoint relations are a subset of the
defining relations of some finitely presented $*$-algebra. Using the fact that the algebra
is stable, we could replace an approximate representation on $\C^*\langle X\rangle$ with the
one that satisfies the self-adjoint relations. In this case, we would like to know the
extent to which this affects the remaining relations. The next result called the
replacement lemma, gives an upper bound on the quality of the approximate representation
obtained in this way. Specifically, by replacing the approximate representation with another approximate representation, whose distance on the generators is known.

\begin{lemma}\label{lem:replacement}
Let $\msA=\C^*\langle S:R\rangle$ and let $R'\subset R$ be a subset of the relations so
that $\msA'=\C^*\langle S:R'\rangle$ is archimedean. There exists constants $K>0$, such
that if $\phi:\C^*\langle S\rangle \to \mcL(H)$ is an $\epsilon$-representation of $\msA$
on a finite-dimensional Hilbert space $H$ with and $\psi:\mbC^*\langle S\rangle \to
\mcL(H)$ is a representation of the quotient $\C^*\langle S:R'\rangle$ with
\begin{equation*}
	\|\phi(s)-\psi(s)\|_f\leq \delta,
\end{equation*}
for all $s\in S$, then $\psi$ is a $(K\delta+\epsilon)$-representation of $\msA$.
\end{lemma}

\begin{proof}
Our proof proceeds in two steps. First, we claim that there exists a constant $M_L>0$ such
that for any monomial $\alpha\in \C^*\langle S\rangle$ of length $L$ we have that
$\|\psi(\alpha)-\phi(\alpha)\|_f\leq M_L\delta$. To begin, we observe that
\begin{equation}\label{eq:mon_len_bd}
	\|\phi(\alpha)-\psi(\alpha)\|_f =\|\phi(s_{i_1})\ldots\phi(s_{i_L})-\psi(s_{i_1})\cdots
	\psi(s_{i_L})\|_f\leq \sum_{k=0}^{L-1} C_0^kC_1^{L-(k+1)}\delta.
\end{equation}
where $C_1=\max_{s\in S}\{\|\phi(s)\|_{op}\}=\kappa_\phi$ and $C_0=\max_{s\in
S}\{\|\psi(s)\|_{op}\}=\vartheta_{\msA'}$ are the largest singular values amongst all the
generators in the image of $\phi$ and $\psi$ respectively. Noting that
\begin{equation}\label{eq:sym_geo}
\sum_{k=0}^{L-1} C_0^kC_1^{L-(k+1)}=\begin{cases}
& \frac{C_1^L-C_0^L}{C_1-C_0}\text{ when $C_1\neq C_0$, and}\\
& LC_0^{L-1}\text{ when $C_0=C_1$},
\end{cases}
\end{equation}
we can take $C=\max\{C_0,C_1\}$ and conclude that the \cref{eq:sym_geo} is bounded by
$M_L:=LC^{L-1}$.
Since each relation $r\in R$ is a finite sum of monomials in the generators $s\in S$ with
complex coefficients, each $r\in R$ can be written as a sum over monomials $\alpha$ of
increasing lengths \begin{equation*} r=\sum_{\alpha} c_\alpha \alpha=\sum_{\ell=0}^N
\left(\sum_{\alpha:len(\alpha)=\ell}c_\alpha \alpha\right).
\end{equation*}
Therefore if $\phi$ is an $\epsilon$-representation, for any $r\in R$, we see that 
\begin{align*}
\|\psi(r)\|_f &\leq \|\psi(r)-\phi(r)\|_f+\epsilon\\
&\leq \sum_{\ell=0}^N \left(\sum_{\alpha:len(\alpha)=\ell}|c_\alpha|\|\psi(\alpha)-\phi(\alpha)\|_f\right)+\epsilon\\
&\leq \sum_{\ell=0}^N \left(\sum_{\alpha:len(\alpha)=\ell}|c_\alpha|\right)M_\ell\delta+\epsilon\\ 
&\leq \sum_{\ell=0}^N \max_\alpha\{|c_\alpha|:len(\alpha)=\ell\} |S|^\ell M_\ell\delta+\epsilon
\end{align*}
the result follows by setting $K=N\max_\ell\{c_\ell |S|^\ell M_\ell \}$, where
$c_\ell=\max_\alpha\{|c_\alpha|:len(\alpha)=\ell\}$. Since $r$ was chosen arbitrarily the result follows.
\end{proof}

If we perform replacement on the stable relations we have that $\delta=O(\epsilon)$ and we
obtain the following important corollary.

\begin{corollary}
If a finitely presented algebra $\C^*\langle S:R\rangle$ contains a subset of relations
$W\subset R$ for which the quotient algebra $\C^*\langle S:W\rangle$ is stable and
archimedean, then \emph{replacing} the $\epsilon$-representation $\phi:\C^*\langle
S\rangle\to \mcL(H)$ by an approximate representation $\psi:\C^*\langle S\rangle \to
\mcL(H)$ that descends to a $*$-homomorphism $\widetilde{\psi}:\C^*\langle S:W\rangle \to
\mcL(H)$ will be an $O(\epsilon)$-representation on $\C^*\langle S:R\rangle$. In
particular, on the relations $r\in R\setminus W$.
\end{corollary}

\begin{remark}\label{rem:bdd_replcmt}
In the proof of \cref{lem:replacement} if the approximate representation $\phi$ is bounded
then we have $C_1=C_0=\vartheta_\msA$. In this case, $L\vartheta_\msA^{L-1}\delta$ bounds
the term in Equation \eqref{eq:mon_len_bd} for each $L>1$ and depends only on the
presentation of $\msA'$. Moreover, unlike \cref{prop:stable_pres}, the bound in
\cref{lem:replacement} is explicit due to the fact that the $*$-polynomials $r\in R$ are
explicitly listed in the presentation so the lengths of the monomials are known
explicitly.
\end{remark}

Another consequence of the replacement lemma is in the case where the operator norms of
the generators in the approximate representations are at most one.

\begin{proposition}\label{prop:bounded_one}
If a finitely presented algebra $\msA=\C^*\langle S:R\rangle$ is archimedean with
$\vartheta_\msA=1$ and $R$ contains a subset of relations $W$ for which the algebra
$\C^*\langle S:W\rangle$ is stable and archimedean, then \emph{replacing} the bounded
$\epsilon$-representation $\phi:\C^*\langle S\rangle\to \mcL(H)$ by an approximate
representation $\psi:\C^*\langle S\rangle \to \mcL(H)$ that descends to a $*$-homomorphism
$\widetilde{\psi}:\C^*\langle S:W\rangle \to \mcL(H)$, will be an
$O(\epsilon)$-representation on $\C^*\langle S:R\rangle$ and the constant depends only on
the presentation of $\msA$.
\end{proposition}

\begin{proof}
The proof is identical to the proof of \cref{lem:replacement}. In particular, since
$\vartheta_\msA\leq 1$ we see that \cref{eq:mon_len_bd} is bounded by $L\delta$ (see
\cref{rem:bdd_replcmt}). Following through the rest of the proof, we see that the
resulting constant $K$ depends only on the presentation of $\msA$.
\end{proof}

Another important class of $*$-algebras/relations we consider in this work are the
\emph{unitary relations} $\{s^*s-1,s^*s-1 :\text{ for all }s\in S\}\subset R$. Suppose
that $\mbC^*\langle S:R\rangle$ is a finitely presented $*$-algebra and $R$ contains the
unitary relations. If $\phi$ is an $\epsilon$-representation of $\C^*\langle S:R\rangle$
and $\phi$ satisfies the unitary relations, then we call $\phi$ a \textbf{unitary
$\epsilon$-representation}. Remark that every unitary approximate representation is
bounded and in particular $\kappa_\phi\leq 1$.

\begin{proposition}\label{prop:unitary}
For any $d\times d$ matrix $X$ with $\|X\|_{op}\leq 1$, there is a unitary $U$ such that
\begin{equation*} \|U-X\|_f\leq \|X^*X-\Id\|_f. \end{equation*}
\end{proposition}

\begin{proof}
Consider the singular value decomposition $X=W\Sigma V$ so that $\Sigma$ is a diagonal
matrix with non-negative singular values $\sigma_j \in [0,1]$ for all $1\leq j\leq d$. Let
$U=WV$ and observe that
\begin{align*}
&\|WV-X\|_f=\|W(\Id-\Sigma)V\|_f=\||\Sigma-\Id|\|_f\\ \leq &
\||(\Sigma-\Id)(\Sigma+\Id)|\|_f=\|\Sigma^2-\Id\|_f=\|V^*(\Sigma U^*U\Sigma-\Id)V\|_f\\
=&\|(V^*\Sigma^* U^*)U\Sigma V-V^*V\|_f=\|X^*X-\Id\|_f,
\end{align*}
since $1-\sigma_j^2=(1-\sigma_j)(1+\sigma_j)\geq 1-\sigma_j$ for all $\sigma_j\in [0,1]$.
\end{proof}

On their own, we see that the unitary relations are $\epsilon$-stable provided the initial
approximate unitary has singular values at most 1 (i.e. it is bounded). If the
largest singular value of $X$ is greater than 1, one could naively take the normalization
$\hat{X}$ with respect to $\|\cdot\|_{op}$, so that $\hat{X}$ has singular value 1.
However, in this case the distance $\|X-W\|_f$ to $W$ will depend on $\|X\|_{op}$, where
$W$ is the unitary in the singular value decomposition of $\hat{X}$. In the worst case,
without any prior bound on $\|X\|_{op}$, we have that $\|X\|_{op}\leq \sqrt{d}\|X\|_f$,
but the resulting stability would not be Hilbert space free (as it would depend on $d$)
even if $\|X\|_f$ is explicitly bounded.

\subsection{Stability for self-adjoint unitaries and PVM algebras}

For our applications to nonlocal games in \cref{sec:nonlocal_games}, we focus our
attention on two important finitely presented algebras: the algebra of self-adjoint
untaries, and the algebra of projective measurement operators (PVMs). Collections of
self-adjoint unitaries arise in the context of quantum strategies for boolean constraint
system nonlocal games, where the measurement operators can be taken without loss of
generality to be boolean ($\pm1$-valued) observables. With this in mind, we define the
finitely presented $*$-algebra of self-adjoint unitaries, denoted by
\begin{equation}\label{eq:sau_alg}
\mcU_n=\C^*\langle x_1,\ldots,x_n: x_i^2-1, x_i^*x_i-1,x_ix_i^*-1,{x_i^*}^2-1\text{ for
all $1\leq i \leq n$}\rangle.
\end{equation}

We remark that $\mcU_n$ is archimedean with radius $\vartheta_{\mcU_n}=1$.

\begin{lemma}\label{lem:sa_un_stab} (\cite{MSZ23}[Lemma 3.10])
If $A$ is a $d\times d$ matrix which satisfies (i) $\|A^2-\Id\|_f\leq\epsilon$, (ii)
$\|A^*A-\Id\|_f\leq \epsilon$, (iii) $\|AA^*-\Id\|_\rho\leq \epsilon$, (iv)
$\|{A^*}^2-\Id\|_f\leq\epsilon$, and (v) $\|A^*-A\|_f\leq \epsilon$, then there exists an
self-adjoint unitary $\tilde{A}$ such that $\|\tilde{A}-A\|_f\leq 2\epsilon$.
\end{lemma}

The idea in the proof is to pick the unitary $\widetilde{A}=\sgn(\frac{A^*+A}{2})$. The result in \cite{MSZ23} show that the above holds in the state-dependent case as
well. We refer the reader to the proof in \cite{MSZ23}. \cref{lem:sa_un_stab} gives the
following immediate result.

\begin{corollary}\label{cor:sau_stab}
The $*$-algebra of self-adjoint unitaries $\C^*\langle
x_1,\ldots,x_n:x_i^2-1,x_i^*x_i-1,x_ix_i^*-1,{x_i^*}^2-1, x_i-x_i^* \text{ for all }1\leq
i\leq n\rangle$ is stable with respect to matrices and $\|\cdot\|_f$.
\end{corollary}

We note that the stability of $\mcU_n$ is Hilbert space free in the sense that there is no
dependence on $d$, additionally it does not depend on $\kappa_\phi$! Another stability
result we require concerns the stability of the group algebra $\C\Z_2^n$, which is
equivalent to the $*$-algebra of self-adjoint unitaries $\mcU_n$ modulo the $*$-ideal
generated by the commutators $[x_i,x_j]=x_ix_j-x_jx_i$ for all $1\leq i \neq j\leq n$.

\begin{lemma}(\cite{Slof19b}[Lemma 24])\label{lem:z2_stab}
There exists a constant $C>0$, such that if $\phi$ is an $\epsilon$-representation of the
group algebra $\C\Z_2^n$ in $M_d(\C)$ then there is a representation $\psi$ of $\C\Z_2^n$
in $M_d(\C)$ such that $\|\psi(s_i)-\phi(s_i)\|_f\leq C\epsilon$, for all $1\leq i \leq
n$. In particular, $\C\Z_2^n$ is stable with respect to $M_d(\C)$ and $\|\cdot\|_f$.
\end{lemma}

We refer the reader to the proof in \cite{Slof19b}. We remark that although they consider
unitary approximate representations (which are bounded), by our result \cref{cor:sau_stab}
stability in that case is sufficient. Since by \cref{lem:replacement} we can first obtain
an $O(\epsilon)$-representation that is self-adjoint and unitary. However, we note that
the constant $C$ could in this case depend on $n$ and $\kappa_\phi$, however, since we
treat $n$ as a fixed parameter it does not affect the stability asymptotically.
Furthermore, if we assume $\phi$ is a bounded approximate representation, then the quality
of the resulting approximate representation depends only on the presentation $\msA$ by
\cref{prop:bounded_one}.

The other important $*$-algebra comes from projective quantum measurements.

\begin{definition}\label{def:pvm_alg}
The PVM algebra $\mcA_{PVM}^{(\mcI,\mcO)}$ is the $*$-algebra:
\begin{equation*}
\C^*\langle \{p_a^i\}_{a\in \mcO,i\in \mcI}: {p_a^i}^2-{p_a}^i, p_a^i{p_a^i}^*-{p_a^i},
{p_a^i}^*p_a^i -{p_a^i}, {{p_a^i}^*}^2-p_a^i ,{p_a^i}^*-p_a^i \text{ for all $a\in \mcO$
and $i\in \mcI$}\rangle
\end{equation*}
Satisfying the additional relations:
\begin{enumerate}[(i)]
\item $p_a^ip_b^i$ for all $a\neq b \in \mcO$ (mutual orthogonality), and
\item $1-\sum_{a\in \mcO}p_a^i$ for each $i\in\mcI$ (completeness).
\end{enumerate}
\end{definition}

Like the algebra of self-adjoint unitaries, this algebra is also archimedean with
$\vartheta_{\mcA_{PVM}^{(\mcI,\mcO)}}=1$. This follows from noting it is a quotient of the
$*$-algebra of positive contractions. We claim that this $*$-algebra is stable with
respect to matrices and $\|\cdot\|_f$. We first collect some results, which are almost
certainly known to experts.

\begin{lemma}\label{lem:proj_stab}
If $A$ is a $d\times d$ matrix which satisfies (i) $\|A^2-A\|_f\leq\epsilon$, (ii)
$\|A^*A-A\|_f\leq \epsilon$, (iii) $\|AA^*-A\|_f\leq \epsilon$, (iv)
$\|{A^*}^2-A\|_f\leq\epsilon$, then there exists an orthogonal projection $\tilde{A}$ such
that $\|\tilde{A}-A\|_f\leq 2(\sqrt{2}+1)\epsilon$.
\end{lemma}

Before we prove \cref{lem:proj_stab} we establish several intermediate claims.

\begin{proposition}\label{prop:con_to_proj}(\cite{KPS18}[Lemma 3.4])
If $C$ is $d\times d$ positive contraction, then there exists a matrix $P$, such that
$P^2=P$ and $P^*=P$, and moreover $\|C-P\|_f\leq 2\sqrt{2}\|C^2-C\|_f$.
\end{proposition}

For a positive contraction $C$, we call the orthogonal projection $P$ in
\cref{prop:con_to_proj} the projective part of $C$ and denote it going forward as
$C_{\{0,1\}}$.

\begin{proposition}\label{prop:pos_to_con}
If $B$ is a $d\times d$ positive (semidefinite) matrix then there exists a positive
contraction $D$ with the property that $\|B-D\|_f\leq \|B^2-B\|_f$.
\end{proposition}

\begin{proof}
Let $\{\lambda_1,\ldots, \lambda_d\}$ be the eigenvalues of $B$, and let $V\subseteq \C^d$
be the image of the joint spectral projections $\{\Pi_{\lambda_i}: \lambda_i\in [0,1]\}$
of $B$ whose corresponding eigenvalues $\lambda_i$ are contained in the interval $[0,1]$
for $1\leq i \leq d$. We define $D$ as the operator which, when restricted to $V$ is equal
to $B$ (i.e. $D|_V=B|_V$). The space orthogonal to $V$ is the image of all spectral
projections $\{\Pi_{\lambda_i}: \lambda_i>1\}$ of $B$ for which the corresponding
eigenvalues are strictly greater than $1$. On $V^\perp$, we define $D$ to be equal to this
projection with eigenvalue $1$ (i.e.~$D|_{V^\perp}$ is the identity matrix). By
construction $D$ a positive contraction. Moreover, the operator $B-D$ has eigenvalues
\begin{equation*}
\mu_i=\begin{cases} 0, \text{ if } 0\leq \lambda_i\leq 1\\ \lambda_i-1, \text{ if } \lambda_i>1,
\end{cases}
\end{equation*}
for $1\leq i \leq d$. Now, if $\lambda_i\in \R$ satisfies $\lambda_i>1$ then we observe that 
\begin{equation*}
\lambda_i^2-\lambda_i= \lambda_i(\lambda_i-1)>(\lambda_i-1).
\end{equation*}
On the other hand, the operator $B^2-B$ has spectrum consisting of the eigenvalues
$\lambda_i^2-\lambda_i$ for $1\leq i \leq d$. The result follows from the calculation
\begin{align*}
\|B-D\|^2_f=\frac{1}{d}\sum_{i=1}^d \mu_i^2 =\frac{1}{d}\sum_{\lambda_i>1}(\lambda_i-1)^2
\leq
\frac{1}{d}\left(\sum_{\lambda_i>1}(\lambda_i^2-\lambda_i)^2+\sum_{\lambda_i\leq1}(\lambda_i^2-
\lambda_i)^2\right)=\|B^2-B\|_f^2,
\end{align*}
where there is equality if $\lambda_i\in \{0,1\}$ for all $1\leq i \leq d$.
\end{proof}

We call the matrix $D$ in \cref{prop:pos_to_con} the \emph{contractive part} of the
positive matrix $B$, and denote it by $B_{[0,1]}$ going forward. We are now ready to prove
\cref{lem:proj_stab}

\begin{proof}[Proof of \cref{lem:proj_stab}]
To begin, it is clear that the matrix $A_{+}=\left(\frac{A^*+A}{2}\right)^2$ is positive
(semi-definite). Furthermore, we observe that $A_+$ is close to $A$ since
\begin{equation}\label{eq:part1}
\|A-\left(\frac{A^*+A}{2}\right)^2\|_f\leq \frac{1}{4}\left(
\|A-A^2\|_f+\|A-A^*A\|_f+\|A-AA^*\|_f+\|A-{A^*}^2\|_f \right)\leq \epsilon,
\end{equation}
by using properties (i)-(iv).
Next, we consider the contractive part of $A_+$, which is defined in
\cref{prop:pos_to_con}, and we denote by $A_{[0,1]}$. If we let
$\{\lambda_1,\dots,\lambda_d\}$ be the eigenvalues of $A_+$, then we observe that
\begin{equation}\label{eq:part2}
\|A_{[0,1]}^2-A_{[0,1]}\|_f^2=\frac{1}{d}\sum_{\lambda_i\in [0,1]}
(\lambda_i^2-\lambda_i)^2 \leq \frac{1}{d}\left(\sum_{\lambda_i\in [0,1]}
(\lambda_i^2-\lambda_i)^2 +\sum_{\lambda_i>0}
(\lambda_i^2-\lambda_i)^2\right)=\|A_+^2-A_+\|_f^2.
\end{equation}
Next, we see that
\begin{equation}\label{eq:part3}
\|A_+^2-A_+\|_f\leq
\frac{1}{4}\left(\|{A^*}^2-A^*\|_f+\|A^*A-A^*\|_f+\|AA^*-A\|_f+\|A^2-A\|_f\right)\leq
\epsilon,
\end{equation}
again using (i)-(iv) and the fact that $\|{A^*}^2-A^*\|_f=\|A^2-A\|_f$ and
$\|A^*A-A^*\|_f=\|AA^*-A\|_f$. For the final step of the proof we let $\tilde{A}$ be the
projective part $A_{\{0,1\}}$ of the positive contraction $A_{[0,1]}$, then by the
triangle inequality along with \cref{prop:con_to_proj}, \cref{prop:pos_to_con}, and
equations \cref{eq:part1}, \cref{eq:part2}, and \cref{eq:part3} we see that
\begin{align*}
\|\tilde{A}-A\|_f&\leq \|A_{\{0,1\}}-A_{[0,1]}\|_f+\|A_{[0,1]}-A_+\|_f+\|A_+-A\|_f\\
&\leq 2\sqrt{2}\|A^2_{[0,1]}-A_{[0,1]}\|_f +\|A_+^2-A_+\|_f+\|A_+-A\|_f\\
&\leq 2\sqrt{2}\epsilon+2\epsilon\\
&=2(\sqrt{2}+1)\epsilon,
\end{align*}
as desired.
\end{proof}

\begin{corollary}
The $*$-algebra of orthogonal projections $\C^*\langle
p_1,\ldots,p_m:p_i^2-p_i,p_i^*p_i-p_i,p_ip_i^*-p_i,{p_i^*}^2-p_i, \text{ for all }1\leq
i\leq m\rangle$ is stable with respect to matrices and $\|\cdot\|_f$.
\end{corollary}

Again we note that this stability is Hilbert space free in that there is no dependence of
$d$ nor on the operator norms of the elements in the approximate representations.

\begin{remark}
If $s$ is the largest singular value of $d\times d$ matrix $A$, and $A$ satisfies
$\|A^2-A\|_f\leq \delta$ and $\|A-A^*\|_f\leq \delta$ for some $\delta>0$, then properties
(i), (ii), (ii), and (iv) in \cref{lem:proj_stab} all hold with $\epsilon=(4s+1)\delta$.
This suggests that the relations $p^*_ip_i-p_i$ and $p_ip_i^*-p_i$ in the orthogonal
projection algebra are fundament in obtaining a stability result that is independent of
the operator norm.
\end{remark}

The stability of the PVM algebra comes from the following lemma.

\begin{lemma}(\cite{Pad23}[Lemma 2.47] \& \cite{Har24}[Remark 2.8])\label{lem:PVM_stab}
There exists a constant $C>0$, such that if $\epsilon>0$, and $A_1,\ldots,A_n$ be positive
contractions in $M_d(\C)$ with the property that (i) $\sum_{i=1}^m\|A_i^2-A_i\|_f\leq
\epsilon$, (ii) $\sum_{1\leq i<j\leq n}\|A_iA_j\|_f\leq \epsilon$, (iii)
$\|\sum_{i=1}^mA_i-\Id\|_f\leq \epsilon$, then there exists a collection of orthogonal
projections $P_1,\ldots, P_n$ such that $P_iP_j=0$ for all $1\leq i\neq j\leq m$,
$\sum_{i=1}^mP_i=\Id$, and $\|A_i-P_i\|_f\leq C\epsilon$ for all $1\leq i \leq n$.
\end{lemma}

We refer to the proof in \cite{Har24}[Lemma 2.7]. Both proofs are based on techniques
presented in \cite{KPS18}[Lemma 3.5]. We note that the constant $C$ in \cref{lem:PVM_stab}
depends exponentially on $m$, however, in our case, we treat $m$ as a fixed parameter so
this is not an issue in this work.

\begin{corollary}\label{cor:pvm_stab}
$\mcA_{PVM}^{(\mcI,\mcO)}$ is stable with respect to $M_d(\C)$ and $\|\cdot\|_f$.
\end{corollary}

\subsection{Rounding and the approximate tracial property}

We now move on to establishing the key technical result in this work. Before we state the
result we define and review a property of state-dependent approximate representations
called the approximate tracial property.

\begin{definition}\label{defn:tracial}
Let $\rho \in \mcL(H)$ be a density matrix and $\msA=\C^*\langle S:R\rangle$ a finitely
presented $*$-algebra. An $(\epsilon,\rho)$-representation $\phi:\C^*\langle S\rangle \to
\mcL(H)$ is \textbf{$\delta$-tracial} if
\begin{equation*}
\|\phi(s)\sqrt{\rho}-\sqrt{\rho}\phi(s)\|_F\leq \delta, \end{equation*} for all $s\in S$.
\end{definition}

We make a few remarks: first, unlike our earlier properties, this property is defined in
terms of the (unnormalized) Frobenius norm $\|\cdot\|_F$. Secondly, if an
$(\epsilon,\rho)$-representation $\phi$ is $(0,\rho)$-tracial, then the linear functional
$\tr(\phi(x) \rho)$ has the tracial (cyclic) property
$\tr(\phi(x)\phi(y)\rho)=\tr(\phi(y)\phi(x)\rho)$ for any $x,y\in X$. Trace linear
functionals, or \emph{tracial states}, play an important role in the representation theory
of finite-dimensional $C^*$-algebras and so it is not surprising to have something
resembling these in the approximate case. The approximately tracial property is a
requirement to state our main rounding lemma.

\begin{definition}\label{def:sau_alg}
A finitely-presented $*$-algebra $\C^*\langle S:R\rangle$ is \textbf{generated by
self-adjoint unitaries} if $R$ contains the relations $W=\{s^2-1,s^*s-1,ss^*-1$,
${s^*}^2-1,s^*-s \text{ for all $s\in S$}\}$. Moreover, we say that $\phi:\C^*\langle
S\rangle \to \mcL(H)$ is a \textbf{self-adjoint unitary approximate representation} if
$\phi(r)=0$ for all $r\in W$.
\end{definition}

In particular, every finitely presented $*$-algebra generated by self-adjoint unitaries is
a quotient of the algebra of self-adjoint unitaries $\mcU_S$. Such algebras are
archimedean with bounded radius at most 1. These algebras are the subject of the following
rounding lemma.

\begin{lemma}\label{lem:rounding}
Let $\msG=\C^*\langle S:R\rangle$ be a finitely-presented $*$-algebra generated by
self-adjoint unitaries. If $H$ is a finite-dimensional Hilbert space, $\varphi:\C^*\langle
S\rangle \to \mcL(H)$ a self-adjoint unitary $(\epsilon,\rho)$-representation of $\msG$,
and $\varphi$ is $\epsilon$-tracial, then there exists a non-zero subspace $\widetilde{H}$
of $H$ and a state-independent $O(\epsilon^{1/2})$-representation $\phi:\C^*\langle
S\rangle \to \mcL(\widetilde{H})$. Moreover, we can choose $\phi$ to be a self-adjoint
unitary approximate representation of $\msG$ on $\widetilde{H}$.
\end{lemma}

Notably, there is no dependence on the dimension of $H$ nor of $\tilde{H}$ in any of the
resulting approximate representations. \cref{lem:rounding} can be seen as an extension of the
proof of Theorem 5.1 in \cite{SV18} from certain groups to certain $*$-algebras. Before we
prove \cref{lem:rounding}, we collect some facts and definitions that are required for the
proof. Let $\chi_I$ be the indicator function for the real interval $I\subseteq \R$. Then
for a self-adjoint operator $T\in \mcL(H)$ and measurable subset $I\subseteq \R$, the
operator $\chi_{I}(T)$ is the spectral projection onto $I\cap spec(T)$, where $spec(T)$ is
the spectrum of $T$. For $\alpha\in \R$, we let $({\geq \alpha})$ denote the interval
$[\alpha,+\infty)$ so that we can express the spectral projection onto $spec(T)\cap
[\alpha,+\infty)$ as the operator $\chi_{\geq \alpha}(T)$. The following result is a
finite-dimensional version of the ``Connes's joint distribution trick'' \cite{Con76}.

\begin{proposition}\label{prop:trick}
If $\lambda$
and $\lambda'$ are positive semi-definite operators on a finite-dimensional
Hilbert space $H$, then
\begin{equation}
\int_0^{+\infty}\|\chi_{\geq
\sqrt{\alpha}}(\lambda)-\chi_{\geq \sqrt{\alpha}}(\lambda')\|_F^2d\alpha\leq
\|\lambda-\lambda'\|_F\|\lambda+\lambda'\|_F.
\end{equation}
\end{proposition}

Rather than giving the proof we direct the reader to the concise proof in
\cite{SV18}[Lemma 5.5]. Readers wishing to see the more general case can consult the
seminal work \cite{Con76}[Lemma 1.2.6]. Along with the above technical result, we require
the following simple result.

\begin{proposition} \label{prop:integral} Let $B$ be a positive
semi-definite operator on a finite-dimensional Hilbert space
\begin{equation*} \int_0^{+\infty}\chi_{\geq
\sqrt{\alpha}}(B)d\alpha=B^2. \end{equation*}
\end{proposition}

Again, a short proof can be found in \cite{SV18}[Lemma 5.6]. In particular, if
$\lambda=\sqrt{\rho}$ for a density matrix $\rho$, then we see that
\begin{equation*}
\int_0^{+\infty}\tr\left(\chi_{\geq
\sqrt{\alpha}}(\lambda)\right)d\alpha=1.
\end{equation*}

In the preliminaries, we defined a \emph{unitary part} of a matrix. However for
\cref{lem:rounding} we need to show that a particular unitary part or a self-adjoint
matrix exists. In particular, there exists a unitary part that when restricted to a
certain subspace is a unitary matrix on that subspace).

\begin{lemma}\label{lem:unitary_res}
Let $X\in M_d(\C)$ be self-adjoint and $P\in M_d(\C)$ be an orthogonal projection.
There exists a unitary part $U$ of $PXP$ which restricts to a unitary on $PM_d(\C)P$.
Moreover, we can pick a $U$ so that $U$ restricted to the image of $P$ is equal to $WV$, where $W\Sigma V$ is the singular value decomposition of $PXP$ in $PM_d(\C)P$.
\end{lemma}

The proof can be found in \cref{sec:parts}.

\begin{proof}[Proof of \cref{lem:rounding}]
Suppose $\msG=\C^*\langle S:R\rangle$ is a finitely presented $*$-algebra and
$\varphi:\C^*\langle S\rangle\to \mcL(H)\iso M_d(\C)$; sending $s_j\mapsto X_j$, is a self-adjoint
unitary $(\epsilon,\rho)$-representation that is $\epsilon$-tracial. Let $\varphi(r)$
represent the image of the relations $r\in R$ in the self-adjoint unitary representatives
$\{X_1,\ldots, X_n\}$ each in $M_d(\C)$.
We begin by showing that there is a non-zero orthogonal projection $P$ on $H$ for which:
\begin{equation}\label{eqn:claim1} \|X_j P-PX_j\|_F= O(\epsilon^{1/2})\tr(P)^{1/2} \text{
for $1\leq j\leq n$, and }\end{equation} \begin{equation}\label{eqn:claim2}
\|\varphi(r)P\|_F= O(\epsilon^{1/2})\tr(P)^{1/2}\text{ for all $r\in R$}. \end{equation}
Both claims follow from \cref{prop:trick}. In particular, for each $1\leq j\leq n$, we see
that \begin{align*} &\int_0^{+\infty}\|X_j\chi_{\geq \sqrt{\alpha}}(\lambda)-\chi_{\geq
\sqrt{\alpha}}(\lambda)X_j\|_F^2d\alpha\\ &=\int_0^{+\infty}\|\chi_{\geq
\sqrt{\alpha}}(\lambda)-X_j^*\chi_{\geq \sqrt{\alpha}}(\lambda)X_j\|_F^2d\alpha \\ &\leq
\|\lambda-X_j^*\lambda X_j\|_F\|\lambda+X_j^*\lambda X_j\|_F\\ &=\|X_j\lambda-\lambda
X_j\|_F\|X_j\lambda+\lambda X_j\|_F\\ &\leq 2 \|X_j\lambda-\lambda X_j\|_F\\ & =
O(\epsilon), \end{align*} by the $\epsilon$-tracial property of the approximate representation
$\varphi$. Additionally, by \cref{prop:integral} we have
\begin{align*} \int_0^{+\infty}\|\varphi(r)\chi_{\geq
\sqrt{\alpha}}(\lambda)\|_F^2d\alpha=\|\varphi(r)\|_\rho^2= O(\epsilon^2),
\end{align*} for each of the relations $r\in R$. Recall that for $\epsilon\leq 1$, we have
$\epsilon^2\leq \epsilon$. Hence, we have that $O(\epsilon^2)=O(\epsilon)$. Therefore,
summing over all relations $r\in R$ we see that \begin{align}\label{eqn:avg}
&\int_0^{+\infty}\bigg(\sum_{j=1}^n \|\chi_{\geq \sqrt{\alpha}}(\lambda)-X_j^*\chi_{\geq
\sqrt{\alpha}}(\lambda)X_j\|_F^2+\sum_{r\in R}\|\varphi(r)\chi_{\geq
\sqrt{\alpha}}(\lambda)\|_F^2\bigg)d\alpha\\ &\leq
O(\epsilon)\int_0^{+\infty}\tr\left(\chi_{\geq \sqrt{\alpha}}(\lambda)\right)d\alpha,
\end{align} holds for any $\alpha\geq 0$. Moreover, each integrand is zero if
$\alpha>\|\lambda\|_{op}^2$ by the definition of $\chi_{\geq \sqrt{\alpha}}(\lambda)$. So
there exists a value $0< \alpha_0\leq \|\lambda\|_{op}^2$ such that if we set
$P:=\chi_{\geq \sqrt{\alpha_0}}(\lambda)$, then $P$ is a non-zero projection and
\begin{equation}\label{eqn:proof1} \sum_{j=1}^n \|X_jP-PX_j\|_F^2 +\sum_{r\in R}
\|\varphi(r)P\|_F^2= O(\epsilon)\tr(P). \end{equation} This in turn, bounds each summand
on the left hand side of equation \cref{eqn:proof1} by $O(\epsilon)\tr(P)$ as all the
terms are positive, establishing that both \cref{eqn:claim1} and \cref{eqn:claim2} hold.
Let $\widetilde{X}_j$ be the unitary part of $PX_jP$ that restricts to a unitary on
$Im(P)$ which we denote as the subspace $\widetilde{H}\subset H$. We note that such unitaries exists by \cref{lem:unitary_res}. We now show that
the following holds in the space $PM_d(\C)P\iso \mcL(\widetilde{H})$.
\begin{enumerate}[(a)] \item 
$\|\widetilde{X}_j-X_jP\|_F=
O({\epsilon}^{1/2})\tr(P)^{1/2}$ for all $1\leq j\leq n$, and \item if $X_{j_1}\cdots X_{j_k}$ for $1\leq
j_1,\cdots,j_k\leq n$ is a word of length $k$, then
\begin{equation}
\|X_{j_1}\cdots X_{j_k}P-\widetilde{X}_{j_1}\cdots \widetilde{X}_{j_k}\|_F=
O({\epsilon}^{1/2})\tr(P)^{1/2}
\end{equation}
 where the constant depends only on $k$, for each $1\leq j\leq n$.
\end{enumerate}
We begin by establishing  (a). From \cref{lem:unitary_res} the matrix $\widetilde{X}_j$ is unitary in $\mcL(\widetilde{H})$ and a unitary part of $PX_jP$. We claim that on this compressed matrix space,
\begin{equation}\label{eqn:proof2}
\|\widetilde{X}_j-PX_jP\|_F= O({\epsilon}^{1/2})\tr(P).
\end{equation}
For the proof, we note that $\|PX_jP\|_{op}\leq \|X_j\|_{op}\leq 1$, hence, if we write $\sigma_i$ for $1\leq i \leq k$ are the singular values of $PXP$, then $\|\widetilde{X}_j-PX_jP\|_F^2=\|W(\Id-\Sigma)V\|_F^2=\sum_{i=1}^k(1-\sigma_i)^2$ by using the singular value decomposition $PX_jP=W\Sigma V$, and the fact that we can take $X_j=WV$ in $PM_d(\C)P$.  Similarly, we rewrite $\frac{1}{2}\|X_jP-PX_j\|_F^2=\tr(P-(PX_jP)^2)=\sum_{i=1}^k(1-\sigma_i^2)$. Futhermore, we have that $0\leq \sigma_i\leq 1$, so it follows that $(1-\sigma_i)^2\leq (1-\sigma_i^2)$, for all $1\leq i \leq k$, and therefore by \cref{eqn:claim2} we see that in $PM_d(\C)P$
\begin{equation}\label{eqn:step1}
\|\widetilde{X}_j-PX_jP\|_F\leq \frac{1}{\sqrt{2}}\|X_jP-PX_j\|_F=O({\epsilon}^{1/2})\tr(P),
\end{equation}
for all $1\leq j\leq n$. Before continuing, we observe that \cref{eqn:step1} also shows that
\begin{equation*}
\|PX_jP-X_jP\|_F=\|PX_jP-X_jP^2\|_F\leq \|PX_j-X_jP\|_F\|P\|_{op}= O({\epsilon}^{1/2})\tr(P)^{1/2}.
\end{equation*}
Hence, by the triangle inequality, for each $1\leq j\leq n$, we see that
\begin{equation}\label{eqn:close}\|\widetilde{X}_j-X_jP\|_F\leq
\|\widetilde{X}_j-PX_jP\|_F+\|PX_jP-X_jP\|_F =
O({\epsilon}^{1/2})\tr(P)^{1/2},
\end{equation}
as desired.

For claim (b), we note that in $PM_d(\C)P$,
$\widetilde{X}_{j}=P\widetilde{X}_{j}$ for any $1\leq j\leq n$, and therefore
\begin{align*}
\|X_{j_1}\cdots X_{j_k}P-\widetilde{X}_{j_1}\cdots
\widetilde{X}_{j_k}\|_F
&\leq O({\epsilon}^{1/2})\tr(P)^{1/2}\\&+\|X_{j_1}\cdots
X_{j_{k-1}}P\widetilde{X}_{j_k}-\widetilde{X}_{j_1}\cdots
\widetilde{X}_{j_{k-1}}\|_F. \end{align*}
The result follows by induction on $k\in \N$.

For the next part of the prof, we show that $(PX_jP)^*(PX_jP)$ is almost the identity on $Im(P)$.
Since $(PX_jP)^*=PX_jP$, we have that
\begin{equation}\label{eqn:proof3}
\|(PX_jP)^2-P\|_F\leq \|X_jPX_j-P\|_F=O({\epsilon}^{1/2})\tr(P)^{1/2}.
\end{equation}

We now conclude the proof by showing that the map $\phi:\C^*\langle S\rangle
\to \mcL(Im(P))$ sending $s_j\mapsto \widetilde{X}_j$ is an
$O(\epsilon^{1/2})$-representation on $Im(P)\subset
H$ with respect to $\|\cdot\|_f$. By construction, we know  that $\widetilde{X}_{j}$
is a unitary on $Im(P)$. To see that each $\widetilde{X}_j$ is close to an involution in $PM_d(\C)P$, we verify
\begin{align*}
\|\widetilde{X}_{j}^2-P\|_f&=\frac{1}{\tr(P)^{1/2}}\|\widetilde{X}_j^2-P\|_F\\
&\leq\frac{1}{\tr(P)^{1/2}}\left(\|\widetilde{X}_j^2-(PX_jP)^2\|_F+\|(PX_jP)^2-P\|_F\right)\\
&\leq\|\widetilde{X}_j^2-(PX_jP)^2\|_f+\frac{1}{\tr(P)^{1/2}}\|(PX_jP)^2-P\|_F\\
&\leq\|\widetilde{X}_j^2-\widetilde{X}_j(PX_jP)\|_f+\|\widetilde{X}_j(PX_jP)-(PX_jP)^2\|_f+O(\epsilon^{1/2})\\
&\leq\|\widetilde{X}_j\|_{op}\|\widetilde{X}_j-PX_jP\|_f+\|\widetilde{X}_j-PX_jP\|_f\|PX_jP\|_{op}+O(\epsilon^{1/2})\\
& \leq \frac{2}{\tr(P)^{1/2}}\|\widetilde{X}_j-PX_jP\|_F+O(\epsilon^{1/2})\\ & \leq
O(\epsilon^{1/2}),
\end{align*}
where we have used the fact that both $\|\widetilde{X}_j-PX_jP\|_F$ and
$\|(PX_jP)^2-P\|_F$ are bounded by $O(\epsilon^{1/2})\tr(P)^{1/2}$ by \cref{eqn:step1} and \cref{eqn:proof2} respectively.

For the remaining relations, we recall the image of each relation under the
approximate representation $\varphi:\C^*\langle S\rangle\to \mcL(\tilde{H})$ is a
$*$-polynomial $\phi(r)=\sum_{M\subset [n]} \gamma_M \prod_{j\in M}\varphi(s_j)$ where
$\gamma_M\in \C$ are coefficients. By the triangle inequality, it suffices to bound each of
the monomials. However, each monomial is bounded by the proofs of claims (a) and (b) earlier.
Hence, for any relation $r\in R$ of $\msG$ we conclude that
\begin{align*}
\|\phi(r)\|_f&=\frac{1}{\tr(P)^{1/2}}\|\phi(r)\|_F\leq
\frac{1}{\tr(P)^{1/2}}\left(\|\phi(r)-\varphi(r)P\|_F+\|\varphi(r)P\|_F\right)\\
&\leq\frac{1}{\tr(P)^{1/2}} \sum_{M\subset [n]} \gamma_M \|\prod_{j\in M}
\widetilde{X}_j-\prod_{j\in M} X_jP\|_F+O(\epsilon^{1/2})\\ &\leq\frac{1}{\tr(P)^{1/2}}
\sum_{M\subset [n]} |\gamma_M| \sum_{j\in M}\|\widetilde{X}_j-X_jP\|_F+O(\epsilon^{1/2})\\
&\leq \sum_{M\subset [n]} |\gamma_M| \sum_{j\in M}O({\epsilon}^{1/2})+O(\epsilon^{1/2})\\
\end{align*}
is $O({\epsilon}^{1/2})$ completing the proof.
\end{proof}

The assumption on the size of $\epsilon$ in \cref{lem:rounding} depends on the relationship
between the quantities in \cref{eqn:avg}. In particular, for the regime of interest with
$\epsilon\rightarrow 0$, we take the larger of the two quantities. One could imagine a
version of \cref{lem:rounding} with an $(\epsilon,\rho)$-representation that is
$\delta$-tracial. The resulting approximate representation would ultimately depend on both
$\epsilon$ and $\delta$. However, in the application of strategies for nonlocal games, we
always obtain an approximate representation where both $\delta$ and $\epsilon$ are determined
by the winning probability of the game, so this technicality does not arise.

An interesting open question is whether you can remove the self-adjoint unitary
assumptions from the state-dependent approximate representation in the hypothesis of
\cref{lem:rounding}. To do this, one would need a state-dependent version of the
replacement lemma. If something like this holds, then one would expect that it could
remove this assumption since the stability of the self-adjoint unitary algebra holds in
the state-dependent case by the result of \cite{MSZ23}. We note that this does not affect
our main results because the approximate representations that come from quantum strategies
are already self-adjoint and unitary.

\section{Near-optimal quantum strategies and nonlocal game algebras}\label{sec:nonlocal_games}

A two-player nonlocal game is a scenario involving two players, Alice and Bob, and a
referee. In the game, Alice (resp.~Bob) receives questions $i\in \mcI_A$ (resp. $j\in
\mcI_B$) from the referee according to a probability distribution $\pi:\mcI_A\times
\mcI_B\to \R_{\geq 0}$. Alice (resp.~Bob) responds to each question with answers $a\in
\mcO_A$ (resp. $b\in \mcO_B$). However, once they receive their questions they are not
permitted to communicate with each other. The goal of the players is to satisfy the rule
predicate $\mcV: \mcO_A\times \mcO_B\times \mcI_A\times \mcI_B\to \{0,1\}$, a function
such that $\mcV(a,b|i,j)=0$ indicates a loss and $\mcV(a,b|i,j)=1$ a win. The description
of the game and the predicate is known to the players before the game begins. The goal of the
players is to maximize their winning probability.

Although communication between the players is not permitted once the game begins, the
players can share a bipartite quantum state. With this resource, the players can make
local quantum measurements on their subsystem to obtain their answers. This allows the
players to employ quantum correlations in their
strategy to win the game. Strategies that employ these quantum correlations are called
quantum (or entangled) strategies, while strategies that use only classical correlations
are called classical strategies.

\begin{definition}\label{def:q_strat}
A \textbf{quantum strategy} $\mcS$ for a nonlocal
game $\mcG$ consists of: \begin{enumerate}[(i)] \item finite-dimensional Hilbert space
$H_A$ and $H_B$, \item collections of orthogonal projections $\{\{P_a^i\}_{a\in
\mcO_A}:i\in \mcI_A\}$ acting on $H_A$, such that $\sum_aP_a^i=\Id_{H_A}$, for all $i\in
\mcI_A$, and a collection of orthogonal projections $\{\{Q_b^j\}_{b\in B}:j\in \mcI_B\}$
acting on $H_B$, such that $\sum_b Q_b^j=\Id_{H_B}$, for all $j\in \mcI_B$, and \item a
quantum state $|\psi\rangle \in H_A\otimes H_B$.\
\end{enumerate}
\end{definition}
These collections of projective-valued measures (PVMs) and a bipartite state correspond to a
quantum strategy $\mcS$ in the sense that they model the correlations
\begin{equation*}
p(a,b|i,j)=\langle \psi|P_{a}^i\otimes Q_{b}^j|\psi\rangle,
\end{equation*} for all outcomes $a\in \mcO_A$, $b\in \mcO_B$, and inputs $ i \in
\mcI_A$, and $j\in \mcI_B$, that can occur in the game. More generally, quantum
correlations can be modelled by positive-operator value-measures (or POVMs) and mixed
states (i.e.~ density operators). However, in this work, we restrict to the class of PVM
strategies with pure states. This is justified by the fact that Naimark's dilation theorem
tells us that any correlation achieved by a POVM on a finite-dimensional Hilbert space can
be achieved by a projective-value-measure (PVM) on a larger (but still) finite-dimensional
Hilbert space. Similarly, a standard purification argument shows that any correlation
achieved with a mixed state can be achieved with a pure state\footnote{In fact, one can
even find a pure state on the same dimensional Hilbert space as the mixed state
\cite{SVW16}.}. Players employing a quantum strategy $\mcS$ for a nonlocal game $\mcG$ win
with probability
\begin{equation}\label{c_value}
\omega(\mcG;\mcS)=\sum_{i,j\in \mcI_A\times
\mcI_B}\pi(i,j)\sum_{a,b\in A\times B}\mcV(a,b|i,j)\langle \psi|P_{a}^i\otimes
Q_{b}^j|\psi\rangle.
\end{equation}
The probability $\omega(\mcG;\mcS)$ is often referred to as the \textbf{value of $\mcG$
under $\mcS$}. The optimal winning probability, called the \textbf{quantum (or
entangled) value} of the game, is the supremum over all quantum strategies and is denoted
by $\omega^*(\mcG)=\sup_\mcS\omega(\mcG;\mcS)$.

\begin{definition}\label{def:approx_strat}
A strategy $\mcS$ for a nonlocal game $\mcG$ is \textbf{perfect} if $\omega(\mcS;\mcG)=1$.
More generally, a strategy is \textbf{optimal} if $\omega(\mcG;\mcS)=\omega^*(\mcG)$. The
notion of near-optimal and near-perfect quantum strategies are natural extensions of the
ideal case. That is, a quantum strategy $\mcS$ is \textbf{$\epsilon$-optimal} if
$|\omega(\mcG;\mcS)-\omega^*(\mcG)|\leq\epsilon$ and is
\textbf{$\epsilon$-perfect}\footnote{We note that our definition of $\epsilon$-perfect
differs from the one used in \cite{SV18}.} if $\omega(\mcG;\mcS)\geq 1-\epsilon$.
\end{definition}

For convenience, we let $p_{ij}(\mcS)$ denote the probability of winning with strategy
$\mcS$ given question pair $(i,j)\in \mcI_A\times \mcI_B$. It is not hard to see that a
quantum strategy $\mcS$ is perfect if and only if $p_{ij}(\mcS)=1$ for all $(i,j)$.
Moreover, any strategy for which $p_{ij}(\mcS)\geq 1-\epsilon$ for all questions $(i,j)$
will be $\epsilon$-perfect. However, we note that the converse is not true. Consider the
case where the distribution of questions is uniform. In this case, a strategy that is
$\epsilon$-perfect implies that $p_{ij}(\mcS)\geq 1-|\mcO_A||\mcO_B|\epsilon$, since a
strategy losing on some question with probability $|\mcO_A||\mcO_B|\epsilon$ could still
win the overall game with probability $1-\epsilon$. Nonetheless, when the questions are
asked uniformly the property of the strategy $\mcS$ being $\epsilon$-perfect and the
property that $p_{ij}(\mcS)\geq 1-\epsilon$ for all $(i,j)$ are equivalent up to a
constant that depends on $\mcG$. Hence, we restrict to the uniform case in this work and
leave the non-uniform case for future work.

\subsection{BCS nonlocal games}\label{sec:bcs_games}

We now focus our attention on the study of Boolean constraint system (BCS) nonlocal games.
BCS games have previously been called ``binary constraint'' nonlocal games, for instance
in \cite{CM14}, however, we prefer the term ``boolean constraint'' as it avoids any
confusion with the subclass of ``2-ary'' constraints. Before we define the BCS nonlocal
game, we review the formal concept of a boolean constraint system.

We adopt the multiplicative convention and associate $-1$ with the boolean $\textsf{TRUE}$
value and $1$ with the boolean $\textsf{FALSE}$. A $k$-ary \textbf{boolean relation}
$\msR$ is a subset of $\{\pm1\}^k$ for $k>0$. Given a set of boolean variables
$V=\{x_1,\ldots,x_n\}$, a \textbf{constraint} $\msC$ is a pair $(\msU,\msR)$ where the the
\textbf{context} $\msU_i$ is the subset of variables $V$ that make up the constraint
$\msC_i$, and $\msR$ is a $k$-ary boolean relation. A satisfying assignment to a
constraint $\msC$ is a function $\phi:V\to \{\pm1\}$ such that $\phi(\msU)\in \msR$. A
\textbf{boolean constraint system} $\msB$ is a pair $(V,\{\msC_i\}_{i=1}^m)$, where $V$ is
a set of variables and $\{\msC_i\}_{i=1}^m$ is a collection of constraints. An assignment
to a BCS $\msB$ is a function $\phi:V\to \{\pm1\}$. The function $\phi$ is a satisfying
assignment if $\phi(\msU_i)\in \msR_i$ for all $1\leq i \leq m$. A BCS is
\textbf{satisfiable} if it has a satisfying assignment. For a single constraint $\msC_i$
we denote the set of satisfying assignments by $sat(\msC_i)=\{\phi:\msU_i \to \{\pm
1\}:\phi(\msS_i)\in \msR_i\}$.

For $z\in \{\pm1\}^k$, and for each $k$ -ary relation $\msR$, we can associate the
\textbf{indicator function} $f_\msR:\{\pm1\}^k \to \{\pm1\}$ that evaluates to $f(z)=-1$
whenever $z\in \msR$ and $1$ otherwise. Given an indicator function $f_\msR$ for a $k$-ary
relation $\msR$ we define the \textbf{indicator polynomial}
\begin{equation}\label{eqn:poly_cons}
\msF_\msR(\msU)=\sum_{z\in \{\pm 1\}^k} f_\msR(z)\prod_{j\in \msU}\frac{(1+z_jx_{j})}{2}=\sum_{M\subseteq \msU_i}\gamma_M \prod_{j\in M} x_j,
\end{equation}
for some coefficients $\gamma_M\in \mathbb{R}$. In other words, the indicator polynomial
for a constraint $\msC$ is a real multilinear polynomial in the context $\msU\subset V$.

For a constraint $\msC$, the indicator polynomial $\msF_\msR(\msU)$ takes the value
$-1$ whenever it is evaluated on a satisfying assignment. Hence, any
propositional formula over the boolean domain has a representation as a multilinear
polynomial. We give a few simple examples:
\begin{example}
For $x,y\in \{\pm 1\}$, the indicator polynomial for $\mathsf{NOT}$ is
$\msF_{\mathsf{NOT}}(x)=-x$, and the $\mathsf{XOR}$ polynomial is given by
$\msF_{\mathsf{XOR}}(x,y)=xy$. The polynomial $\mathsf{AND}(a_1,a_2)=a_1a_2$ becomes the
$\pm1$ values polynomial
$\widetilde{\mathsf{AND}}(x_1,x_2)=\frac{1}{2}(1+x_1+x_2-x_1x_2)$.
\end{example}

Given a BCS $\msB$ we can define a two-player \textbf{BCS nonlocal game} $\mcG(\msB)$. In
the game, Alice receives a constraint $\msC_i$ for some $i\in \mcI_A$ with $|\mcI_A|=m$
and replies with an assignment $\phi$ to $\msC_i$. Meanwhile, Bob receives a single
variable $x_j$ for $j\in \mcI_B$ with $|\mcI_B|=n$ and replies with an assignment to the
single variable $\varphi \in \{\pm1\}$. They win the game if $\phi$ satisfies $\msC_i$ and
their assignment are consistent, that is $\phi(x_j)=\varphi$ for all $x_j\in \msU_i$,
otherwise they lose. The probability distribution for the game describes the probability
of Alice being given the constraint $\msC_i$ and Bob the variable $x_j$. If $\phi$ is a
satisfying assignment for $\msB$, then the players can always win with probability 1 by
employing the strategy where they both use $\phi$ to obtain their answers.

For a BCS nonlocal game, the questions and answers are indexed by $\mcI_A=[m]$,
$\mcI_B=[n]$, $\mcO_A=sat(\msC_i)$, and $\mcO_B=\{\pm1\}$. In this work, we assume the
distribution of questions is uniform on the variables and constraints. A quantum strategy
$\mcS$ for a BCS game consists of PVMs $\{\{P_{a}^i\}_{a\in sat(\msC_i)} :1\leq i \leq
m\}$ and $\{\{Q_b^j\}_{b \in \Z_2}:1\leq j \leq n\}$. That is Alice has a projective
measurement system over satisfying assignments to each constraint, and Bob has a projective
measurement system over the $\pm1$-assignments to each variable.

As mentioned earlier, a convenient way to analyze quantum strategies for BCS nonlocal
games is in terms of the \emph{bias} rather than the winning probability. Recall that the
bias of a strategy $\mcS$ is the probability of winning minus the probability of losing.
When considering the bias of a quantum strategy for a BCS nonlocal game, the relevant $\pm
1$-valued observables are:
\begin{equation}\label{eqn:obs_def}
Y_{ij}=\sum_{a\in sat(\msC_i)} {a_j} P_{a}^i,\text{ and } X_j=\sum_{b \in\Z_2}b Q_{b}^j.
\end{equation}
To see why, recall that the winning probability on inputs $(i,j)$ is given by
\begin{equation}
p_{ij}(\mcS)=\sum_{b\in \Z_2}\sum_{\substack{a\in sat(\msC_i)\\a_j=b}}p(a,b|i,j).
\end{equation}
Let $\beta_{ij}(\mcS)=\langle\psi|Y_{ij}\otimes X_j|\psi\rangle$ be the bias on question
$(i,j)$ with quantum strategy $\mcS$ and observe that
\begin{align*} \langle\psi|Y_{ij}\otimes X_j|\psi\rangle
=&\sum_{a\in sat(\msC_i)}\sum_{b\in \Z_2}a_{j}\cdot b\langle\psi|P_{a}^i\otimes
Q_b^j|\psi\rangle\\&
=2\left[\sum_{b\in \Z_2}\sum_{\substack{a\in sat(\msC_i)\\a_j=b}}p(a,b|i,j)\right]-1\\
&=2p_{ij}(\mcS)-1.
\end{align*} It follows that $\beta_{ij}(\mcS)=1$ if and
only if $p_{ij}(\mcS)=1$, hence we can characterize perfect quantum strategies in terms of
these $\pm1$-valued observables.

\begin{proposition}\label{prop:perfect_bcs}
Let $\mcS$ be an $\epsilon$-perfect quantum strategy for a BCS game $\mcG(\msB)$ with
vector state $|\psi\rangle\in H_A\otimes H_B$. Then there is a collection of
$\pm1$-valued observables $\{Y_{ij}\}_{i,j\in [m]\times [n]}$ in $\mcL(H_A)$ and
$\{X_j\}_{j=1}^n$ in $\mcL(H_B)$ such that $[Y_{ij},Y_{ik}]=0$ for all $j,k$, and $\langle
\psi|Y_{ij}\otimes X_j|\psi\rangle\geq 1-O(\epsilon)$, whenever $\mcV(a,b,i,j)=1$, for all
$i,j\in [m]\times [n]$.
\end{proposition}

\begin{proof}
Take $Y_{ij}$ and $X_j$ to be the operators defined in \cref{eqn:obs_def} for all $1\leq i
\leq m$ and $1\leq j\leq n$. By construction each $Y_{ij}$ and $X_j$ is a self-adjoint
unitary. Moreover, the commutation relations $[Y_{ij},Y_{ik}]=0$ hold, since each $Y_{ij}$
and $Y_{ik}$ is a $\Z_2$-linear combinations of orthogonal projections
$P_{\underline{a}}^i$ for all $1\leq i \leq m$. Lastly, if there is an $\epsilon$-perfect
strategy then $p_{ij}(\mcS)\geq 1-\epsilon$ and
\begin{equation}
 \langle\psi|Y_{ij}\otimes X_j|\psi\rangle=2p_{ij}(\mcS)-1\geq 1-2\epsilon,
\end{equation}
whenever $\mcV(a,b,i,j)=1$, for all $i,j\in [m]\times [n]$ as desired.
\end{proof}

Cleve and Mittal observed that the collection of mutually commuting $\pm1$-valued
observables $X_{j}$ derived from a perfect quantum strategy for a BCS nonlocal game
satisfy the multilinear polynomials in \cref{eqn:poly_cons}, see \cite{CM14}[Theorem 1].
Hence, perfect quantum strategies for a BCS game correspond to a \emph{matrix-valued}
satisfying assignments to the BCS.

\begin{definition} Given a BCS $\msS$ a \emph{matrix-valued satisfying assignment} for
$\msB$ is a collection of mutually commuting $\pm1$-valued observables $\{X_1,\ldots,
X_n\}$ such that \begin{equation} \msF_{\msR_i}(X_{j_1},\ldots,X_{j_k})=-\Id,
\end{equation} for all constraints $(\msU_i,\msR_i)$, where $j_1,\ldots,j_k\in \msU_i$ is
an abuse of notation to indicate the index of the variables appearing in $\msU_i$, for
$1\leq i \leq m$. Moreover, it is not hard to see that we can write the constraint
polynomial as a difference of projections \begin{align}
\msF_{\msR_i}(\msU_i)&=\sum_{\phi(\msU_i)\notin \msR_i}\prod_{j\in
\msU_i}\frac{\Id+\phi(x_j)X_j}{2}-\sum_{\phi(\msU_i)\in \msR_i}\prod_{j\in
\msU_i}\frac{\Id+\phi(x_j)X_j}{2}\\ &=2\sum_{\phi(\msU_i)\notin \msR_i}\prod_{j\in
\msU_i}\frac{\Id+\phi(x_j)X_j}{2}-\Id. \end{align} For convenience, we will denote the
orthogonal projections \begin{equation} \Pi_{\phi,i}(X_j)=\prod_{j\in
\msU_i}\frac{\Id+\phi(x_j)X_j}{2}, \end{equation} where $\phi:\msU_i\to \{\pm1\}$ and
$1\leq i \leq m$. \end{definition}

It is not hard to see that the converse is also true. Specifically, given a matrix
satisfying assignment of $\pm1$-valued observables $\{X_1,\ldots, X_n\}$ of dimension $d$,
the BCS nonlocal game can be played perfectly using a strategy where Bob employs the
observables $\Id\otimes X_j$, Alice employs the observables $Y_{ij}\otimes
\Id={X_j}^\top\otimes \Id$ for all $1\leq i \leq m$ and $1\leq j\leq n$, and they share a
maximally entangled state $|\tau\rangle=\sum_{i=1}^d|i\rangle \otimes |i\rangle$. We refer
the reader to the details in \cite{CM14,Ji13}.

The BCS algebra is based on the observation that polynomial constraint relations
abstractly define perfect strategies for BCS nonlocal games.

\begin{definition}\label{def:BCS_algebra}
The \textbf{BCS algebra} $\mcB$ of a boolean
constraint system $\msB$ is the self-adjoint unitary algebra $\mcU_n$ (defined in \cref{eq:sau_alg}) subject to the additional relations: \begin{enumerate}[(1)]
\item $\msF_{\msR_i}(\msU_i)=-1$ for all $1\leq i \leq m$, and
\item  $x_{j}x_\ell=x_\ell x_{j}$ whenever $x_j,x_{\ell}\in \msU_i$, for all $1\leq i\leq m$.
\end{enumerate}
Here, each $\msF_{\msR_i}(V)=\msF_{\msR_i}(\msU_i)$ is the multilinear indicator
polynomial for the constraint $\msC_i$ with context $\msU_i$.
\end{definition}

In a perfect strategy for BCS nonlocal games, one can show that Alice and Bobs' measurement
operators must be the same on the support of the state (up to transposition in the Schmidt
basis of the state vector). For self-adjoint operators, the transpose is equivalent to the
conjugate (in that basis). The following lemma shows that a similar result holds in the
approximate case.

\begin{lemma}\label{lem:frob_com}
Let $X$ and $Y$ be self-adjoint unitary operators on a finite-dimensional Hilbert space $H$
and $|\psi\rangle\in H\otimes H$. Then,
\begin{equation}
\langle \psi|X\otimes Y|\psi\rangle\geq 1-O(\epsilon),
\end{equation} if and only if 
\begin{equation}
\|Y\lambda-\lambda \overline{X}\|_F\leq O(\epsilon^{1/2}),
\end{equation}
where $\lambda$ is the square root of the reduced density
matrix $\rho$ for the state $|\psi\rangle$.
Moreover, in any case where the above
holds we also have that \begin{enumerate}[(i)] \item $\|Y\lambda-\lambda
Y\|_F\leq O(\epsilon^{1/2})$, and \item
$\|\overline{X}\lambda-\lambda\overline{X}\|_F\leq O(\epsilon^{1/2})$.
\end{enumerate}
\end{lemma}

The proof can be found in \cite{SV18}[Proposition 5.4]. By combining this result with
\cref{prop:perfect_bcs} we obtain a simple corollary.

\begin{corollary}\label{cor:close_eps}
Let $\mcS$ be a quantum strategy for a BCS game $\mcG(\msB)$ presented in terms of $\pm
1$-valued observables $ \{Y_{ij}\}_{i,j=1}^{m,n}$, and $\{X_{j}\}_{j=1}^n$, along with a
maximally entangled state $|\tau\rangle\in H_A\otimes H_B$. If
\begin{equation}
\epsilon^2\geq \|Y_{ij}^\top-X_j\|_f^2=2(1-\langle\tau|Y_{ij}\otimes X_j|\tau\rangle)
\end{equation}
for all $1\leq i \leq m$, and $1\leq j\leq n$, then $p_{ij}(\mcS)\geq1-
\frac{\epsilon^2}{4}$ for all question pairs $(i,j)$ and $\mcS$ is
$\epsilon^2/4$-perfect.
\end{corollary}

We now show that Bob's (or analogously Alice's) operators in any $\epsilon$-perfect
strategy to the BCS game is a state-dependant approximate representation of the BCS
algebra $\mcB(\mcG)$. In the remainder of the section, we assume that the state
$|\psi\rangle$ is fully supported on the spaces $H_A$ and $H_B$. In this case, the
$\rho$-seminorm is a norm on the space of linear operators on $H_B$ (or $H_A$).

\begin{proposition}\label{prop:BCS_rho_approx}
Let $\rho=\lambda^*\lambda$ on $H_B$ be the reduced density matrix of the state
$|\psi\rangle\in H_A\otimes H_B$ on $H_B$. If
$(\{Y_{ij}\}_{i,j=1}^{m,n},\{X_{j}\}_{j=1}^n,|\psi\rangle\in H_A\otimes H_B)$ is an
$\epsilon$-perfect strategy for the BCS game $\mcG$, then $\{X_j\}_{j=1}^n$ is an
$(O(\epsilon^{1/2}),\rho)$-representation of the BCS algebra $\mcB(\mcG)$. Moreover, the
approximate representation is $O(\epsilon^{1/2})$-tracial.
\end{proposition}

\begin{proof}
A priori each $X_j$ is a self-adjoint unitary, so all that remains is to establish that
\begin{enumerate}[(a)]
\item $\|\msF_{\msR_i}(V)+\Id\|_\rho\leq
O(\epsilon^{1/2})$ for all $1\leq i\leq m$, and \item $\|X_\ell X_j-X_jX_\ell\|_\rho\leq
O(\epsilon^{1/2})$ whenever $X_\ell,X_k\in \msU_i$, for all $1\leq i\leq m$.
\end{enumerate}
For (a), abusing some notation we let $\msU_i$ be the variables that appear in the $i$th
constraint polynomial $\msF_{\msR_i}$. For convenience, we define
$Z_{ij}:=\overline{Y_{ij}}$ for all $1\leq i\leq m$ and $1\leq j\leq n$. Then, it is not
too hard to show that
\begin{equation}\label{eqn:monomials}
\left\|\prod_{j\in \msU_i} X_j\lambda-\lambda\prod_{j\in \msU_i}Z_{ij}\right\|_F\leq
\sum_{j\in \msU_i}\|X_j\lambda-\lambda Z_{ij}\|_F,
\end{equation}
for any $M\subset V$.

Since Alice's observables are a matrix-valued satisfying assignment, we have that
$\msF_{\msR_i}(Z_{i_1},\ldots,Z_{i_k})=-\Id$. Thus, since $\|X_j\lambda-\lambda
Z_{ij}\|_F\leq O(\epsilon^{1/2})$ for each $j \in M$ by \cref{lem:frob_com}, we can deduce
that \begin{align*} \|\msF_{\msR_i}(\msU_i)+\Id\|_\rho &\leq \left\| \sum_{S\subset
\msU_i}\gamma_M \prod_{j\in M}X_j\lambda-\lambda( -\Id)\right\|_F\\ &\leq \sum_{S\subset
\msU_i}|\gamma_M| \left\|\prod_{j\in M}X_j\lambda-\lambda\prod_{j\in
M}Z_{ij}\right\|_F=O(\epsilon^{1/2}), \end{align*} as desired. To see that (b) holds, we
observe that
\begin{equation}\label{eqn:2mon}
\|X_kX_j-X_jX_k\|_\rho\leq
\|X_kX_j\lambda-\lambda Z_{ij}Z_{ik}\|_F+ \|X_jX_k\lambda-\lambda
Z_{ik}Z_j\|_F,
\end{equation}
since the variables $\{Z_{j}:j\in \msU_i\}$ all commute by virtue of Alice
employing a valid quantum strategy. Now we apply the inequality from
\cref{eqn:monomials} to \cref{eqn:2mon} and conclude that (b) holds by
\cref{lem:frob_com}. Lastly, the approximate tracial property $\|X_j\lambda-\lambda
X_j\|^2_F= O(\epsilon)$ for all $1\leq j\leq n$ follows directly from the second statement
of \cref{lem:frob_com}.
\end{proof}

We immediately obtain the following corollary.

\begin{corollary}\label{cor:BCS}
If the state $|\psi\rangle\in H_A\otimes H_B$ in the $\epsilon$-perfect strategy $\mcS$
for a BCS game $\mcG$ is maximally entangled, then the operators $\{X_j\}_{j=1}^n$ are a
state-independent $O(\epsilon^{1/2})$-representation of the BCS algebra $\mcB$ on $H_B$.
\end{corollary}

For arbitrary states we employ the rounding lemma (\cref{lem:rounding}) to obtain our main
result.

\begin{proposition}\label{prop:res1}
If $\mcS$ is an $\epsilon$-perfect strategy for a BCS nonlocal game $\mcG$, then
restricted to a non-zero subspace of $H_B$, Bob's measurement operators are
state-independent $O(\epsilon^{1/4})$-representation of the BCS algebra $\mcB(\mcG)$.
\end{proposition}

\begin{proof}
It follows from \cref{prop:BCS_rho_approx} that any $\epsilon$-perfect strategy for a BCS
game $\mcG$ with reduced density matrix $\rho$, gives a state-dependent
$(O(\epsilon^{1/2}),\rho)$-representation of the BCS algebra $\mcB(\mcG)$ that is
$O(\epsilon^{1/2})$-tracial. Since the operators in any strategy are $\pm1$-valued
observables the approximate representation already satisfies the self-adjoint unitary
relations exactly. Hence, the result follows by applying \cref{lem:rounding}, which
results in an $O(\epsilon^{1/4})$-representation of the BCS algebra in the
$\|\cdot\|_f$-norm.
\end{proof}

To prove \cref{thm:second}(1) we show that each $\epsilon$-representations of the BCS
algebra can be used to obtain near-perfect strategies for the corresponding BCS nonlocal
game provided the players share a maximally entangled state. For this result, we appeal to
the stability of the group algebra $\C\Z_2^k$.

\begin{proposition}\label{prop:BCS_alg}
If $\phi$ is a bounded $\epsilon$-representation of the BCS algebra $\mcB$ on a
finite-dimensional Hilbert space $H_B$, then there is a $O(\epsilon^2)$-perfect strategy
to the BCS game $\mcG$ using the maximally entangled state $|\tau \rangle\in H_B\otimes
H_B$.
\end{proposition}

\begin{proof}
Let $\varphi$ be an $\epsilon$-representation of $\mcB(\mcG)$ on $H_B$. For a fixed
constraint $\msC_i$, consider the $\epsilon$-representation restricted to the context
$\msU_i$. On these variables $\varphi$ is an $\epsilon$-representation of
$\C\Z_2^{|\msU_i|}$. By \cref{lem:z2_stab} the algebra $\C\Z_2^k$ is stable. In
particular, there exists an exact representation $\psi_i$ of $\C\Z_2^{|\msU_i|}$ such that
$\|\varphi(x_j)-\psi_i(x_j)\|_f\leq C\epsilon$ for all $j\in \msU_i$. Define the
$\pm1$-observables $W_{ij}=\psi_i(x_j)$ for all $j\in \msU_i$. The collection of
observables $\{W_{ij}\}_{j\in \msU_i}$ correspond to measuring an assignment
$\phi:\msU_i\to \{\pm 1\}$, however, it may not be a satisfying assignment for $\msC_i$.
In particular, if $\Pi_{\phi,i}(W_{ij})=\prod_{j\in
\msU_i}\tfrac{1}{2}(\Id+\phi(x_j)W_{ij})$ then the the projection onto unsatisfying
assignments $\sum_{\phi(\msU_i)\notin\msR_i}\Pi_{\phi,i}(W_{ij})$ may be non-zero.
However, this happens if and only if an irreducible block of $W_{ij}$ with corresponding
weight $\nu(x_j) \in \{\pm 1\}$ does not agree with $\phi(x_j)$ (i.e.~ if
$\phi(x_j)=-\nu(x_j)$). Hence, multiplying each unsatisfying irreducible block in each
$W_{ij}$ by $-1$, leaving the other blocks untouched, we obtain a new $\pm1$-observable
$Y_{ij}$. Repeating this for each $j\in \msU_i$ gives us a new collection of observables
for which $\sum_{\phi(\msU_i)\notin\msR_i}\Pi_{\phi,i}(Y_{ij})=0$. More precisely, define
the new collection of observables \begin{equation}
Y_{ij}^\top=\sum_{\phi(\msU_i)\notin\msR_i}
(\Id-\Pi_{\phi,i}(W_{ij}))W_{ij}-\sum_{\phi(\msU_i)\notin\msR_i}\Pi_{\phi,i}(W_{ij})W_{ij},
\end{equation} for all $j\in \msU_i$. In addition to these observables being a satisfying
assignment for $\msC_i$, we have that \begin{align*}
\|W_{ij}-Y_{ij}^\top\|_f&=\|W_{ij}-\sum_{\phi(\msU_i)\notin \msR_i}
(\Id-\Pi_{\phi,i}(W_{ij}))W_{ij}+\Pi_{\phi,i}(W_{ij})W_{ij}\|_f\\
&=\|\Id-\sum_{\phi(\msU_i)\notin\msR_i}
(\Id-\Pi_{\phi,i}(W_{ij}))+\Pi_{\phi,i}(W_{ij})\|_f\\
&=\|2\sum_{\phi(\msU_i)\notin\msR_i}\Pi_{\phi,i}(W_{ij})\|_f\\ &\leq
2\|\sum_{\phi(\msU_i)\notin\msR_i}\Pi_{\phi,i}(W_{ij})-\sum_{\phi(\msU_i)\notin\msR_i}\Pi_{\phi,i}(X_j))\|_f+\|2\sum_{\phi(\msU_i)\notin\msR_i}\Pi_{\phi,i}(X_{j})\|_f\\
&\leq
2\sum_{\phi(\msU_i)\notin\msR_i}\|\Pi_{\phi,i}(W_{ij})-\Pi_{\phi,i}(X_j)\|_f+\epsilon\\
&\leq 2\sum_{j\in \msU_i}\|W_{ij}-X_j\|_f+\epsilon\\ &\leq 2\sum_{j\in
\msU_i}\|\psi_i(x_j)-\varphi(x_j)\|_f+\epsilon\\ &=(2|\msU_i|C+1)\epsilon, \end{align*}
for all $j\in \msU_i$. This means that the number of assignments in $W_{ij}$ that we had
to correct (by multiplying by $-1$) was relatively small, provided we have $\epsilon$
small. Now, consider the strategy consisting of observables $\{X_j\}_{j=1}^n$,
$\{Y_{ij}\}_{i=1,j=1}^{m,n}$ and a maximally entangled state $|\tau\rangle\in H_B \otimes
H_B$. By \cref{cor:close_eps} the strategy is $((Cn+1)\epsilon)^2$-perfect, since
$\|Y_{ij}^\top-X_j\|_f\leq \|Y_{ij}^\top-W_{ij}\|_f+\|W_{ij}-X_j\|_f=2(Cn+1)\epsilon$, for
all $1\leq j \leq n$, and $1\leq i \leq m$.
\end{proof}

\begin{remark}
The only caveat in the case where $\phi$ is not bounded a priori is that the resulting
approximate representations could depend on $\kappa_\phi$. If this quantity is independent
of $d$ then the result will still be Hilbert space free in this case as well.
\end{remark}

As a corollary, we obtain the final result of this section.

\begin{corollary}
For any $\epsilon$-perfect quantum strategy $\mcS$ for a BCS nonlocal game there is an
$O(\epsilon^{1/2})$-perfect quantum strategy $\tilde{\mcS}$ using a maximally entangled
state $|\tilde{\psi}\rangle$, such that each measurement in $\tilde{\mcS}$ is at most
$O(\epsilon^{1/4})$-away from the measurement in $\mcS$ with respect to $\|\cdot\|_f$ on
the local support of $|\tilde{\psi}\rangle$ on $H_B$.
\end{corollary}

\subsection{Synchronous nonlocal games}\label{sec:synch_games}

Synchronous games are class of nonlocal games introduced in \cite{PSSTW16}. In a
synchronous nonlocal game $\mcG$, the set of questions and answers are the same for each
player, i.e.~ $\mcI_A=\mcI_B=\mcI$ and $\mcO_A=\mcO_B=\mcO$. Additionally, in a
synchronous game, the players lose whenever they give different answers to the same
question. This latter property is called the \emph{synchronous condition}. Like BCS
nonlocal games there is a finitely presented $*$-algebra we can associate with each
synchronous game. Unlike BCS algebras, the associated game algebra is typically presented
in terms of collections of orthogonal projections:

\begin{definition}\label{def:synch_alg}
The \textbf{synchronous algebra} $\mcA(\mcG)$ of a synchronous game $\mcG$ is a quotient
of the PVM algebra $\mcA_{PVM}^{(\mcI,\mcO)}$ (see \cref{def:pvm_alg}) satisfying the
additional rule relations: for all $(a,b,i,j)\in \mcO^2\times \mcI^2$ if
$\mcV(a,b,i,j)=0$, then ${p_{a}^{i}}{p_{b}^{j}}=0$.
\end{definition}

The class of synchronous algebras is well studied in the literature, and typically in the
context of the more general commuting-operator strategies, see for instance
\cite{PT15,DP16,HMPS19}. It is not hard to see that any representation of the synchronous
on a finite-dimensional Hilbert space $H$, along with a maximally entangled state gives a
perfect quantum strategy for the associated synchronous case. The other direction is less
straightforward, but also true as a result of \cite{PT15}[Theorem 5.5]. Hence, a synchronous nonlocal
game has a perfect quantum strategy $\mcS$, if and only if Bob's PVM measurement operators
$\{\{Q_a^i\}_{a\in\mcO}:i\in \mcI\}$ restricted to the support of the reduced density
matrix $\rho_B$ forms a representation of the synchronous algebra associated with $\mcG$.
To analyze near-perfect strategies for synchronous games, we start by mentioning the
following simple lemma.

\begin{lemma}(\cite{Slof11}[Lemma 4.1]) \label{lem:frob_norm_state}
If $E$ and $F$ are self-adjoint operators on Hilbert spaces $H_A$ and $H_B$ respectively,
and $\lambda=\rho^{1/2}$ is the reduced density matrix of a pure state $|\psi\rangle\in
H_A\otimes H_B$ on $H_B$, then \begin{equation*}\|(E\otimes \Id-\Id\otimes
F)|\psi\rangle\|=\|\lambda \overline{E}-F\lambda\|_F, \end{equation*} where $\overline{E}$
is the entry-wise conjugate taken with respect to the basis of eigenvectors for
$\lambda\in \mcL(H_B)$. Therefore if $\|(E\otimes \Id-\Id\otimes
F)|\psi\rangle\|_{H_A\otimes H_B}\leq \epsilon$ then $\| E\lambda-\lambda
\overline{F}\|_F\leq \epsilon$.
\end{lemma}

\begin{proposition} \label{prop:proj_ATP}
Suppose $\mcG$ is a synchronous nonlocal game with a uniform distribution on questions.
If $\{\{E_a^i\}_{a\in O}:i\in \mcI\}$ and $\{\{F_a^i\}_{a\in O}:i\in \mcI\}$ are Alice and
Bob PVM's from an $\epsilon$-perfect strategy for $\mcG$, using the state $|\psi\rangle$,
then $\|E_a^i\lambda-\lambda \overline{F_{a}^i}\|_F\leq O(\epsilon^{1/2})$ for all $i\in
\mcI$, and $a\in \mcO$.
\end{proposition}

\begin{proof}
If $\mcS$ is an $\epsilon$-perfect strategy then for any pair of questions $(i,j)$ we have\begin{equation*}
\sum_{a,b\in \mcO\times
\mcO}\mcV(a,b|i,j)p(a,b|i,j)\geq 1-n^2\epsilon.
\end{equation*}
In particular $\mcV(a,b,i,j)=0$ whenever
$a\neq b$, so for $\mcS$ to be $\epsilon$-perfect it must be that $\sum_{a\neq b}\langle
\psi|E_a^i\otimes F_b^i|\psi\rangle\leq n^2\epsilon$ holds, for all $1\leq i\leq m$. Hence, we see that
\begin{align}
\|(E_a^i\otimes \Id-\Id\otimes
F_a^i)|\psi\rangle\|^2=\sum_{b\neq a}\langle \psi|E_a^i\otimes
F_b^i|\psi\rangle+\sum_{a'\neq a}\langle \psi|E_{a'}^i\otimes
F_a^i|\psi\rangle \leq 2n^2 \epsilon,
\end{align}
and the result follows from \cref{lem:frob_norm_state}.
\end{proof}

A synchronous nonlocal game is \textbf{symmetric} if $\mcV(b,a|j,i)=0$ whenever
$\mcV(a,b|i,j)=0$, for all $a,b\in \mcO$, $i,j\in \mcI$. Although not every synchronous
predicate is symmetric, it is not hard to show that every perfect quantum strategy is
symmetric. That is, if $\mcV(a,b,i,j)=0$ and $\mcS$ is a perfect quantum strategy, then
$p(a,b|i,j)=p(b,a|j,i)=0$, see \cite{HMPS19}]Corollary 2.2].

\begin{proposition}\label{lem:near_sym}
Let $\mcS=(\{\{E_a^i\}_{a\in \mcI}:i\in \mcI\},\{\{F_a^i\}_{a\in \mcO}:i\in
\mcI\},|\psi\rangle)$ be a $\epsilon$-perfect strategy for a synchronous nonlocal game
$\mcG$, where $\rho=\lambda^2$ is the reduced density matrix of $|\psi\rangle$ on $H_B$.
If $\mcV(a,b,i,j)=0$ for $a,b\in \mcO$, $i,j\in \mcI$, then $p(b,a|j,i)=p(a,b|i,j)+
O(\epsilon^{1/2})$.
\end{proposition}

\begin{proof}
We observe that
\begin{align*}
p(b,a|i,j)&=\langle \psi|E_b^j\otimes F_a^i|\psi\rangle\\ &=\langle \psi| (E_b^j\otimes
\Id)( \Id \otimes F_a^i) |\psi\rangle\\ &= \langle \psi|(E_b^j\otimes \Id-\Id \otimes
F_b^j)^*(\Id\otimes F_a^i - E_a^i\otimes \Id)|\psi\rangle+\langle\psi| E_a^i\otimes
F_b^j|\psi\rangle\\ &+ \langle \psi|(\Id \otimes F_b^j)(\Id\otimes F_a^i - E_a^i\otimes
\Id)|\psi\rangle+ \langle \psi|(E_b^j\otimes \Id-\Id \otimes F_b^j)(E_a^i\otimes
\Id)|\psi\rangle\\ &\leq \|E_b^j\lambda-\lambda {F_b^j}^T\|_F \|\lambda {E_b^j}^T-
F_b^j\lambda\|_F+\|E_b^j\lambda-\lambda {F_b^j}^T\|_F+\|\lambda {E_b^j}^T-
F_b^j\lambda\|_F\\&+p(a,b|i,j), \end{align*}
the result follows from \cref{prop:proj_ATP}.
\end{proof}

\begin{proposition}\label{prop:synch_approx_rho}
If $\mcS=(\{\{E_a^i\}_{a\in \mcO}:i\in \mcI\},\{\{F_a^i\}_{a\in \mcO}:i\in
\mcI\},|\psi\rangle)$ is an $\epsilon$-perfect strategy for a synchronous game $\mcG$,
where $\rho=\lambda^2$ is the reduced density matrix of $|\psi\rangle$ on $H_B$, then
$\{\{F_a^i\}_{a\in \mcO}:i\in \mcI\}$ is an $O(\epsilon^{1/4})$-representation of the
synchronous game algebra in $\mcL(H_B)$ with respect to the state-induced semi-norm
$\|\cdot\|_\rho$. Moreover, the $O(\epsilon^{1/4})$-representation is
$(O(\epsilon^{1/2}),\lambda)$-tracial.
\end{proposition}

\begin{proof}
Since $\mcS$ is $\epsilon$-perfect, whenever
$\mcV(a,b|i,j)=0$ we have that 
\begin{equation*}
p(a,b|i,j)=\langle \psi|E_a^i\otimes F_b^j|\psi\rangle
=\tr(\overline{E_a^i}\lambda F_b^j\lambda)\leq n^2\epsilon
\end{equation*}
hence, by \cref{lem:near_sym} we have $p(b,a|j,i)=\tr(\overline{E_b^j}\lambda
F_a^i\lambda)= O(\epsilon^{1/2})$.
From Cauchy-Schwarz we deduce that
\begin{align*}
\|F_a^iF_b^j\|^2_\rho&=\|F_a^iF_b^j\lambda\|_F^2\\
&=\tr(\lambda F_b^jF_a^iF_b^j\lambda)\\ &=\tr(\lambda F_b^jF_a^i(F_b^j\lambda-\lambda
\overline{E_b^j}))+\tr(\lambda F_b^jF_a^i\lambda\overline{E_b^j})\\
&\leq \|F_a^iF_b^j\lambda\|_F\|F_b^j\lambda-\lambda
\overline{E_b^j}\|_F+ \tr(F_a^i\lambda\overline{E_b^j}\lambda F_b^j))\\
&=  \|F_a^iF_b^j\lambda\|_F\|F_b^j\lambda-\lambda
\overline{E_b^j}\|_F+ \tr(F_a^i\lambda\overline{E_b^j}(\lambda F_b^j-\overline{E_b^j}\lambda))+
\tr(F_a^i\lambda\overline{E_b^j}^2\lambda)\\
&\leq \|F_a^iF_b^j\lambda\|_F\|F_b^j\lambda-\lambda
\overline{E_b^j}\|_F+\|\overline{E_b^j}\lambda
F_a^i\|_F\|\lambda F_b^j-\overline{E_b^j}\lambda\|_F +\tr(\overline{E_b^j}\lambda F_a^i\lambda)\\
&\leq\|F_b^j\lambda-\lambda
\overline{E_b^j}\|_F+\|\lambda F_b^j-\overline{E_b^j}\lambda\|_F +\tr(\overline{E_b^j}\lambda F_a^i\lambda)\\
&=O(\epsilon^{1/2}),
\end{align*}
using \cref{prop:proj_ATP}. Since $\{\{F_a^i\}_{a\in \mcO}: i \in \mcI\}$ is a PVM, each
$F_a^i$ is an orthogonal projection and $\sum_{a\in \mcO} F_a^i=\Id_{H_B}$ for each $i\in
\mcI$, and so the remaining relations in \cref{def:synch_alg} hold automatically. The fact
that the approximate representation is $O(\epsilon^{1/2},\lambda)$-tracial follow from the
second part of \cref{lem:frob_com}, by letting $X_a^i=\Id-2E_a^i$ and $Y_a^i=\Id-2F_a^i$
we obtain self-adjoint unitaries satisfying $\|Y_a^i\lambda-\lambda
\overline{X}_a^i\|_F\leq O(\epsilon^{1/2})$ for all $i\in \mcI$ and $a\in \mcO$,
completing the proof.
\end{proof}

We leave it as an open problem whether this bound on the state-dependent approximate
representation is tight in the degree of $\epsilon$. It seems plausible that
\cref{lem:near_sym} could be improved to $O(\epsilon)$, which would lead to an
$O(\epsilon^{1/2})$-representation in \cref{prop:synch_approx_rho}. Despite this issue, we
obtain the immediate corollary when the state in the strategy is maximally entangled.

\begin{corollary}\label{cor:synch}
Let $\mcS$ be an $\epsilon$-perfect strategy for a synchronous nonlocal game $\mcG$. If
the state $|\psi\rangle$, from the strategy $\mcS$, is maximally entangled with Schmidt
rank equal to $dim(H_B)$, then Bob's measurement operators $\{\{F_a^i\}_{a\in O}:i \in
I\}$ are an $O(\epsilon^{1/4})$-representation of the synchronous algebra on $\mcL(H_B)$
with respect to $\|\cdot\|_f$.
\end{corollary}

Unlike in the BCS case, we cannot directly apply \cref{lem:rounding} to
\cref{prop:synch_approx_rho} and obtain a state-independent approximate representation.
This is because the presentation of the synchronous algebra is in terms of orthogonal
projections and not self-adjoint unitaries. In the next subsection, we show that there is
an alternative presentation of the synchronous algebras as the BCS algebra of a certain
BCS built from the original synchronous game.

\subsection{The SynchBCS algebra of a synchronous nonlocal game}

Although synchronous and BCS games may initially appear different they are essentially the
same. First off, there is a synchronous version of any BCS game by considering the game
where Alice and Bob each receive a constraint $\msC_i$ and $\msC_j$ and must reply with
satisfying assignments. In this ``constraint-constraint'' version of the BCS game, the
players win perfectly if and only if their assignment to all variables in the intersection
of the contexts match. This synchronous version of a BCS game is well-known. In
particular, the authors of \cite{KPS18} employ this idea to construct a synchronous
nonlocal game for which there is a $*$-homomorphism from the synchronous algebra of a
\emph{linear} BCS game to the corresponding $C^*$-algebra of the corresponding
\emph{solution group} of the linear system.

In this work, we focus on the other direction. We consider a ``constraint-variable''
version of a synchronous nonlocal game which we call the SynchBCS game. Given a
synchronous nonlocal game $\mcG$ the SynchBCS game associated with $\mcG$ is the BCS game
where:
\begin{itemize}
\item[-] for each question $i\in \mcI$ and answer $a\in \mcO$, we add a 
$\{\pm1\}$-valued variable $z_a^i$, and
\item[-] in the
synchronous game, whenever $\mcV(a,b|i,j)=0$ we add the constraint
$\widetilde{\text{AND}}(z_a^i,z_b^j)=1$, and
\item[-] to ensure that each ${z_{a}^{i}}$ comes from a single measurement\footnote{For
any $i\in \mcI$ the subset of ${z_{a}^{i}}$'s are jointly measurable, and exactly one
of them outputs a $-1$.}, we add the constraint
$\widetilde{\text{XOR}}_{a\in O}(z_{a}^i)=-1$ for each question $i\in \mcI$.
\end{itemize}
This last constraint prevents two different $-1$'s from each question while ensuring at
least one $-1$ output is given for each input $i$. In this SynchBCS game, the players can
receive any of these constraints and they must reply with a satisfying assignment to the
variables in the context. The distribution on these constraints is informed by the
distribution of questions in the original synchronous game. Since this is a BCS game, we
can consider the corresponding BCS algebra associated with any synchronous game $\mcG$
through this transformation.

\begin{definition}\label{def:synchBCS}
The \textbf{SynchBCS algebra} $\msB(\mathcal{G})$ of the synchronous nonlocal game $\mcG$
is a quotient of the self-adjoint unitary algebra $\mcU_{\mcI\times \mcO}$ (see
\cref{eq:sau_alg}), with self-adjoint unitary generators $\{{z_{a}^{i}}: (i,a)\in
\mcI\times \mcO\}$, satisfying the additional relations:
\begin{enumerate}
\item $\text{$\widetilde{\mathsf{AND}}$}({z_{a}^{i}},z_{b}^{j})=1$, whenever
$\mcV(a,b|i,j)=0$ for all $i,j\in \mcI$ and $a,b\in \mcO$,
\item $\prod_{a\in \mcO}{z_{a}^{i}}=-1$, each $i\in \mcI$,
\item ${z_{a}^{i}}{z_{a'}^{i}}={z_{a'}^{i}}{z_{a}^{i}}$ for all pairs $a,a'\in \mcO$, and each $i\in \mcI$.
\end{enumerate} \end{definition}

By construction, the finite-dimensional representations of the SynchBCS algebras give
quantum satisfying assignment to the associated SynchBCS nonlocal game. Hence, there is a
BCS nonlocal game for each synchronous nonlocal game. In \cite{Fri20,Gol21}, the authors
showed that the synchronous LCS game algebra is isomorphic\footnote{As mentioned, one of
the directions was also established in \cite{KPS18}. We also note that although they use
``BCS'' in their title, their results are only for BCS games with linear constraints.} to
the synchronous algebra of projections. We establish the following complementary result.

\begin{proposition}\label{prop:iso}
The synchronous game algebra $\mcA(\mcG)$ is $*$-isomorphic to the SynchBCS algebra $\msB(\mcG)$.
\end{proposition}

\begin{proof}
We begin by describing the $*$-homomorphism $\phi:\mcA(\mcG)\arr \msB(\mcG)$. Define the
function on the generators ${p_{a}^{i}}\mapsto (1-{z_{a}^{i}})/{2}$. This function extends
to a $*$-homomorphism $\phi$ on $\mcA(\mcG)$. We now check that it descends to a
$*$-homomorphism to $\msB(\mcG)$. First, we note that each $\phi({p_{a}^{i}})$ is an
orthogonal projection since ${z_{a}^{i}}^*={z_{a}^{i}}$, and ${z_{a}^{i}}^2=1$, for all
$a\in \mcO$. The $\widetilde{\mathsf{AND}}$ relation together with the relation
${z_{a}^{i}}z_{a'}^i=z_{a'}^i{z_{a}^{i}}$ for all $a,a'\in \mcO$, implies that
$1-{z_{a}^{i}}-z_{b}^j+{z_{a}^{i}}z_{b}^j=0$ whenever $\mcV(a,b|i,j)=0$, and thus
$$\phi({p_{a}^{i}})\phi({p_{b}^{i}})=\frac{(1-{z_{a}^{i}})}{2}\frac{(1-z_{j,b})}{2}=0$$ is
satisfied whenever $\mcV(a,b|i,j)=0$.

For each $i\in \mcI$ with $|\mcO|=n$, observe that the unit $1$ can be expanded as the sum
of indicator polynomials in the variables ${z_{a}^{i}}$, giving us
\begin{align}\label{eqn:id_func} 1=\sum_{(e_1,\ldots,e_n)\in \{\pm1\}^n}\prod_{a\in
\mcO}\frac{(1+e_a{z_{a}^{i}})}{2}. \end{align} However, upon enforcing the orthogonality
relations we note that $\prod_{a\in \mcO}\frac{(1+e_a{z_{a}^{i}})}{2}=0$, whenever there
is a pair $a,a'\in \mcO$ with $e_a=e_{a'}=-1$. Thus, there are only two cases we need to
consider.

The first is when $e_a=1$ for all $a\in \mcO$. In this case, we have the term
\begin{align}\label{eqn:badterm} \prod_{a\in
\mcO}\frac{(1+z_a^i)}{2}&=\frac{1}{2^n}\left(\sum_{S\subseteq [n]}\prod_{a\in
S}z_a^i\right). \end{align}

Recalling that the rule predicate relation shows $\prod_{a\in \mcO}{z_{a}^{i}}=-1$, we
observe that \begin{equation}\label{eqn:miracle} \prod_{a\in S}z_a^i+\prod_{a\in
[n]\setminus S}z_a^i=0, \end{equation} for any $S\subseteq [n]$, by recalling that each
${z_a^i}^2=1$ by the self-adjoint unitary relations. It follows that 
\cref{eqn:badterm} is $0$ because each subset $S\subseteq [n]$ is in bijection with its
complementary subset $S^c=[n]\setminus S$, and so by equation \cref{eqn:miracle} each term
with an $S$ product cancels out with the term for $S^c$ product.

In the other case, the remaining terms are those with exactly one $a\in \mcO$ with
$e_a=-1$. For this case, let $f_a^i=(1-{z_{a}^{i}})/2$ and $f_{a'}^i=(1-z_{a'}^i)/2$ and
observe that $f_{a}^i$ and $f_{a'}^{i}$ are self-adjoint orthogonal projections with
$f_a^i f_{a'}^i=0$, and therefore $f_a^i(1-f_{a'}^i)=f_a^i-f_a^if_{a'}^i=f_{a}^i$. Then,
noting $1-f_{a'}^i=(1+z_{a'}^i)/2$, it follows that $f_a^i\left[ \prod_{a'\neq a}
(1-f_{a'}^i)\right]=f_a^i$. These being the only remaining terms in
\cref{eqn:id_func}, we see that \begin{equation*} 1=\sum_{a\in \mcO}
\frac{(1-{z_{a}^{i}})}{2}\prod_{a'\neq a}\frac{(1+z_{a'}^i)}{2}=\sum_{a\in \mcO}
\frac{(1-{z_{a}^{i}})}{2}=\sum_{a\in \mcO} \phi({p_{a}^{i}}), \end{equation*} for all
$i\in \mcI$, as desired.

On the other hand, consider the $*$-homomorphism $\varphi:\msB(\mcG)\arr \mcA(\mcG)$
defined by extending the function ${z_{a}^{i}}\mapsto (1_\mcA-2{p_{a}^{i}})$. Recalling
that ${p_{a}^{i}}{p_{b}^{i}}=0$ for all $a\neq b$, we see that \begin{align*} \prod_{a\in
\mcO} \varphi({z_{a}^{i}})=\prod_{a\in \mcO}(1-2{p_{a}^{i}})= \sum_{S\subset
\mcO}(-2)^{|S|}\prod_{a\in S}{p_{a}^{i}}= 1_\mcA+(-2)\sum_{a \in \mcO}{p_{a}^{i}}=-1_\mcA,
\end{align*} by recalling the completeness relation in $\mcA(\mcG)$. Now, if
$\mcV(a,b|i,j)=0$ then we have that ${p_{a}^{i}}{p_{b}^{j}}=0$, hence \begin{align*}
\text{$\widetilde{\mathsf{AND}}$}(\varphi({z_{a}^{i}}),\varphi({z_{b}^{i}}))&=\frac{1}{2}\big(1_\mcA+(1+2{p_{a}^{i}})+(1_\mcA+2{p_{b}^{j}})\\
&-(1_\mcA+2{p_{a}^{i}})(1_\mcA+2{p_{b}^{j}})\big)\\ &=\frac{1}{2}\left(2\cdot
1_\mcA+2{p_{a}^{i}}+2{p_{b}^{j}}-2{p_{a}^{i}}-2{p_{b}^{j}}-4{p_{a}^{i}}{p_{b}^{j}}\right)\\
&=1_\mcA. \end{align*} Lastly, since $V(a,b|i,i)=0$, it follows that
$\varphi({z_{a}^{i}})\varphi({z_{b}^{i}})=(1_\mcA+2{p_{a}^{i}})(1_\mcA+2{p_{b}^{i}})=1_\mcA+2{p_{a}^{i}}+2{p_{b}^{i}}=\varphi({z_{b}^{i}})\varphi({z_{a}^{i}})$
for all $a\neq b$ as desired.

It remains to show that $\varphi$ and $\phi$ are mutual inverses. Observe that,
\begin{align*}
\varphi\left(\phi({p_{a}^{i}})\right)&=\varphi\left(\frac{1-{z_{a}^{i}}}{2}\right)=\frac{1}{2}\left(\varphi(1)-\varphi({z_{a}^{i}})\right)=\frac{1}{2}\left(1_\mcA-(1_\mcA-2{p_{a}^{i}})\right)={p_{a}^{i}}.
\end{align*} Similarly, \begin{align*}
\phi\left(\varphi({z_{a}^{i}})\right)&=\phi(1_\mcA-2{p_{a}^{i}})=\phi(1_\mcA)-2\phi({p_{a}^{i}})=1-2\frac{(1-{z_{a}^{i}})}{2}={z_{a}^{i}},
\end{align*} thus $\varphi\circ \phi=id_{\mcA(\mcG)}$ and $\phi\circ
\varphi=id_{\msB(\mcG)}$, and the result follows. \end{proof}

Recall from an $\epsilon$-perfect strategy for a BCS nonlocal game with an arbitrary
state, our rounding result will give us an $O(\epsilon^{1/4})$-representation of the BCS
algebra in the $\|\cdot\|_f$-norm. To apply our rounding result in the synchronous algebra
case we need to ensure that $\epsilon$-representations of the synchronous algebra
$\mcA(\mcG)$ gives us $O(\epsilon)$-representation of the SynchBCS algebra in the
$\rho$-norm under the isomorphism in \cref{prop:iso}. We remark that we cannot directly
apply \cref{prop:stable_pres} here because the approximate representation is not
state-independent!

\begin{proposition}\label{prop:synch_to_bcs}
If $\psi$ is an $(\epsilon,\rho)$-representation of $\mcA(\mcG)$, then $\phi\circ \psi$ is
a $(O(\epsilon),\rho)$-representation of $\msB(\mcG)$, where the $*$-homomorphism
$\phi:\mcA(\mcG)\to \msB(\mcG)$ is described in \cref{prop:iso}.
\end{proposition}

\begin{proof}[Proof of \cref{prop:synch_to_bcs}]
It is straightforward to show that if $\|{F_a^i}^2-F_a^i\|_\rho\leq \epsilon$
and $\|{F_a^i}^*-F_a^i\|_\rho\leq \epsilon$, then $\|{X_a^i}^2-\Id\|_\rho\leq
4\epsilon$, and similarly $\|{X_a^i}^*-X_a^i\|_\rho\leq 2\epsilon$. So all that remains to show
that relations (1), (2), and (3) hold approximately as well.
If $\mcV(a,b|i,j)=0$ then $\|F_a^iF_b^j\|_\rho\leq \epsilon$, hence, we set
${X_i^a}=\Id-2F_a^i$. Note that
$\|\widetilde{\mathsf{AND}}({z_{a}^{i}},z_{b}^{j})-1\|_\rho=\|1-{z_{a}^{i}}-z_{b}^{j}+{z_{a}^{i}}z_{b}^{j}\|
_\rho$.
Then
\begin{equation}\label{eq:rem_rel}
\| \Id-{X_{a}^{i}}-X_{b}^{j}+{X_{a}^{i}}X_{b}^{j}\|_\rho =
4\left\| \frac{(\Id-{X_{a}^{i}})}{2}\frac{(\Id-X_{b}^{j})}{2}\right\|_\rho=
4 \|F_{a}^iF_b^j\|_\rho\leq 4\epsilon.
\end{equation}

Lastly,

\begin{align*} \left\|\prod_{a\in A}X_{a}^i+\Id\right\|_\rho=&\left\|\prod_{a\in
A}(\Id-2F_a^i)+\Id\right\|_\rho\\ =&\left\|\sum_{\alpha\subset
A}(-2)^{|\alpha|}\prod_{a\in \alpha}F_a^i+\Id\right\|_\rho\\
=&\left\|\sum_{|\alpha|=1}(-2)F_a^i+2\Id+\sum_{|\alpha|>1}(-2)^{|\alpha|}\prod_{a\in
\alpha}F_a^i\right\|_\rho\\ \leq&
2\|\Id-\sum_{a}F_a^i\|_\rho+\sum_{|\alpha|>1}2^{|\alpha|}\|\prod_{a\in
\alpha}F_a^i\|_\rho\\ \leq& 2\epsilon+\sum_{|\alpha|>1}2^{|\alpha|}\prod_{a''\in
\alpha\setminus\{a,a'\}}\|F_{a''}^i\|_{op}\|F_{a'}^iF_a^i\|_\rho\\ \leq&
2\epsilon+\sum_{|\alpha|>1}|2|^{|\alpha|}C^{|\alpha|-1}\|F_a^iF_{a'}^i\|_\rho\\
\leq &2\epsilon+2^{2|A|}C^{|A|-1}\epsilon\\ = &O(\epsilon), \end{align*}
where $C$ is the constant that bounds the operator norm of each $F_{a}^i$. Now, we ensure
that the commutation relation holds for all $a\neq a'$, $i\in \mcI$
\begin{align*}
\|X_a^iX_{a'}^i-X_{a'}^iX_a^i\|_\rho&=\|(\Id-2F_a^i)(\Id-2F_{a'}^i)-(\Id-2F_{a'}^i)(\Id-2F_a^i)\|_\rho\\
&\leq 4(\|F_a^iF_{a'}^i\|_\rho+\|F_{a'}^iF_{a}^i\|_\rho)\\ &= O(\epsilon),
\end{align*} as desired.
\end{proof}

\begin{remark}
Although \cref{prop:synch_to_bcs} works for arbitrary $(\epsilon,\rho)$-representations,
for our applications the approximate representation of $\mcA(\mcG)$ in the $\rho$-norm
already satisfy exactly many of the relations in the synchronous algebra
\cref{def:synch_alg}. This is because they begin as projective measurements over
satisfying assignments for some quantum strategy. Moreover, if ${X_a^i}=\Id-2F_a^i$ is the
$\pm1$-valued observable assigned to the orthogonal projection onto outcome $(i,a)$ then
under the isomorphism in \cref{prop:iso} the collection of observables
$\{X_a^i\}_{(i,a)\in \mcI\times \mcO}$ generate a self-adjoint unitary
$O(\epsilon^{1/2})$-representation of the SynchBCS algebra $\msB(\mcG)$. It is clear that
the relations self-adjoint unitary relations, as well as relations (2), and (3) in
\cref{def:synchBCS} hold exactly in this approximate representation due to the fact $\mcS$
is a representation of the PVM algebra. Therefore it only remains to check that relation (1)
in \cref{def:synchBCS} holds approximately, which follows from \cref{eq:rem_rel}. Lastly,
we see that since the projections satisfy the property in \cref{prop:proj_ATP}, their
corresponding observables satisfy the hypothesis of \cref{lem:frob_com}.
\end{remark}

\begin{corollary}\label{cor:synch_bcs_rounded}
If $\mcS$ is an $\epsilon$-perfect synchronous strategy for a
synchronous nonlocal game $\mcG$, then the corresponding
$(O(\epsilon^{1/4}),\rho)$-representation of the SynchBCS algebra is
$O(\epsilon^{1/2})$-tracial.
\end{corollary}

\cref{cor:synch_bcs_rounded} allows us to apply our rounding result to synchronous
algebras by considering the approximate representations of the SynchBCS algebra. In
particular, the representation coming from a near-perfect strategy for the synchronous game
yields a self-adjoint unitary state-dependent approximate representation of the SynchBCS
algebra. Hence, we can apply \cref{lem:rounding} to the SynchBCS algebra to obtain a
self-adjoint unitary state-independent approximate representation of the SynchBCS algebra.
Then, by the isomorphism, we can return the state-independent approximate representation of
the SynchBCS algebra to obtain a state-independent approximate representation of the
synchronous algebra. We summarize the result in the following proposition.

\begin{proposition}\label{prop:res2}
If $\mcS$ is an $\epsilon$-perfect strategy for a synchronous nonlocal game $\mcG$, then
restricted to a non-zero subspace of $H_B$, Bob's measurement operators are an
$O(\epsilon^{1/8})$-representation of the synchronous algebra $\mcA(\mcG)$.
\end{proposition}

\begin{proof}
The synchronous case is similar to the BCS case. By \cref{prop:synch_approx_rho} any
$\epsilon$-perfect strategy for a synchronous game $\mcG$ with reduced density matrix
$\rho$ gives a state-dependent $(O(\epsilon^{1/4},\rho)$-representation of the synchronous
algebra $\mcA(\mcG)$ that is $O(\epsilon^{1/2})$-tracial. By \cref{cor:synch_bcs_rounded},
this state-dependent approximate representation is a $O(\epsilon^{1/2})$-tracial
$(O(\epsilon^{1/4},\rho)$-representation of the synchBCS algebra $\msB(\mcG)$ that is
exact on the self-adjoint unitary relations. By applying \cref{lem:rounding}, we obtain a
$O(\epsilon^{1/8})$-representation in the $\|\cdot\|_f$-norm of the synchBCS algebra on a
subspace of $H_B$. Finally, by the $*$-isomorphism described in \cref{prop:iso} combined
with \cref{prop:stable_pres}, we obtain a state-independent
$O(\epsilon^{1/8})$-representation of the synchronous algebra.
\end{proof}

Our final task in this subsection is to show that approximate representations of the
synchronous algebra is close to near-perfect quantum strategies.

\begin{proposition}\label{prop:synch_alg}
If $\phi$ is a bounded $\epsilon$-representation of the synchronous algebra $\mcA(\mcG)$
on a Hilbert space $H$, then there is $O(\epsilon^2)$-perfect synchronous strategy using a
maximally entangled state.
\end{proposition}

\begin{proof}
Suppose that $\phi$ is an $\epsilon$-representation of the synchronous algebra on a
Hilbert space $H$. By the stability of the PVM algebra $\mcA_{PVM}^{(\mcI,\mcO)}$
\cref{cor:pvm_stab}, there exists a constant $C_0\geq 0$ and orthogonal projections
$\{\{\Pi_a^i:a\in A\}:i\in I\}$, such that $\sum_{a} {\Pi_a^i}=\Id$ for all $i\in I$,
${\Pi_a^i}{\Pi_b^i}=0$ for $a\neq b$, for all $i\in I$, and
$\|{\Pi_a^i}-\phi(p_a^i)\|_f=C_0\epsilon$, for all $i\in I$, $a\in A$. Moreover, by
replacing the almost-PVMs with the genuine PVMs the replacement lemma shows that new PVMs
still approximately satisfy the rule relations. More precisely, by \cref{lem:replacement}
there exists a constant $C_1>0$ such that $\|{\Pi_a^i}{\Pi_b^j}\|_f=C_1\epsilon$ whenever
$\mcV(a,b|i,j)=0$. Now, consider the quantum strategy where Alice employs the PVMs
$\{\{P_a^i={\Pi_a^i}^\top:a\in A\}:i\in I\}$, Bob employs PVMs $\{\{Q_b^j={\Pi_a^i}:b\in
A\}:j\in I\}$, and they use a shared maximally entangled state $|\tau\rangle \in H\otimes
H$. Given this strategy, the probability of losing on question pair $(i,j)$ is at most
\begin{align*}
\sum_{a,b:\mcV(a,b|i,j)=0}\langle\psi|{P}_{a}^i\otimes
{Q}_b^j|\psi\rangle &=\sum_{a,b:\mcV(a,b|i,j)=0}\widetilde{\tr}({P_{a}^i}^\top Q_b^j)\\
&=\sum_{a,b:\mcV(a,b|i,j)=0}\|\Pi_{a}^i \Pi_b^j\|^2_f\\
&\leq|\mcO|^2C_1^2\epsilon^2.
\end{align*}
It follows that the strategy is $O(\epsilon^2)$-perfect.
\end{proof}

\begin{corollary}\label{cor:synch2}
For any $\epsilon$-perfect
quantum strategy $\mcS$ for a synchronous nonlocal game, there is an
$O(\epsilon^{1/4})$-perfect quantum strategy $\tilde{\mcS}$ using a maximally entangled
state $|\tilde{\psi}\rangle$, such that each measurement in $\tilde{\mcS}$ is at most
$O(\epsilon^{1/8})$-away from the measurement in $\mcS$ with respect to $\|\cdot\|_f$ on
the local support of $|\tilde{\psi}\rangle$ on $H_B$.
\end{corollary}

\subsection{XOR nonlocal games}

Unlike the case of BCS and synchronous games, XOR nonlocal games do not admit
non-classical perfect quantum strategies \cite{CHTW10}. However, in many cases, there are
quantum strategies for XOR games that can achieve higher winning probabilities than the
best classical strategies \cite{Slof11,CHTW10}. Nevertheless, XOR games have an affiliated
finitely presented $*$-algebra $\msC(\mcG)$, called the XOR algebra, for which optimal
quantum strategies correspond to representations.

An XOR game is a nonlocal game where Alice and Bob are given questions $i\in\mcI_A=[m]$
and $j\in \mcI_B=[n]$ according to a probability distribution $\pi(i,j)$, and they respond
with 1-bit answers $a\in \mcO_A=\{0,1\}$ and $b\in \mcO_B=\{0,1\}$. The rule predicate for the game is
determined by the XOR of the answer bits. For any XOR game, we can describe the predicate
by an $m\times n$ $\{0,1\}$-matrix $T$ with entries $(T)_{i,j}=t_{ij}$, so that
\begin{equation}
\mcV(a,b|i,j)=\begin{cases} &1,\text{ if $a\oplus b=t_{ij}$}, \\ &0, \text{ otherwise}
\end{cases}.
\end{equation}
Letting $w_{ij}=(-1)^{t_{ij}}\pi(i,j)$, we obtain the \textbf{cost matrix} $W$ of an XOR
game. Given the cost matrix of a game, one can conveniently express the \textbf{bias} of a
strategy $\mcS$ consisting of $\pm1$-valued observables $\{Y_1,\ldots, Y_m\}$ for Alice,
$\{X_1,\ldots, X_n\}$ for Bob, and a vector state $|\psi\rangle\in H_A\otimes H_B$, as
\begin{equation}\label{eqn:xor_bias}
\beta(\mcG;\mcS)=\sum_{i=1,j=1}^{m,n} w_{ij}\langle \psi|Y_i\otimes X_j|\psi\rangle.
\end{equation}
The supremum over all quantum strategies $\mcS$ gives the \textbf{optimal bias} denoted by
$\beta_q(\mcG)$ for the XOR game $\mcG$.

A result of Tsirelson implies that the quantum value of an XOR game can be computed using
a semi-definite program \cite{CHTW10,Whe06,CSUU08}. Unlike for BCS games, this
characterization often makes computing the entangled value of an XOR game computationally
tractable. Also, using the formulation of Tsirelson it was shown in \cite{Slof11} that if
a strategy $\mcS$ for an XOR game is optimal, then the measurement observables satisfy the
following relation
\begin{equation}\label{eqn:xor_rig}
\sum_{j=1}^n w_{ij}X_j\lambda=r_i\lambda\overline{Y}_i,
\end{equation}
for all $i\in [m]$. Here, the collection $\{r_i\}_{i\in [m]}$ are the \textbf{marginal row
biases} associated with the questions, and $\lambda$ is the square root of the reduced
density matrix of the strategy quantum state $|\psi\rangle$ on $H_B$. From this relation,
one can define the XOR-algebra associated to an XOR nonlocal game $\mcG$ in terms of
abstract relations resembling equation \cref{eqn:xor_rig} along with the $\pm 1$-valued
observable relations, see \cite{Slof11}.

\begin{definition}\label{def:xor_alg}
Let $\mcG$ be an XOR game with an $m\times n$ cost matrix $W$, and marginal row biases
$\{r_i\in \R:i\in[m]\}$. The \textbf{XOR algebra} $\mcC(\mcG)$ is a quotient of the
self-adjoint unitary algebra $\mcU_n$ (see \cref{eq:sau_alg}), subject to the additional
relations:
\begin{equation}
\left(\sum_{j=1}^n w_{ij}x_j\right)^2=r_i^2\cdot 1\text{ for all $1\leq i\leq m$}
\end{equation}
\end{definition}

The characterization of optimal strategies in terms of the semi-definite program was also
applicable in the approximate setting, in particular, it implies the following
approximate rigidity result.

\begin{theorem}[Theorem 3.1 in \cite{Slof11}]\label{thm:xor}
For every XOR game $\mcG$ there exists a collection of constants $r_i\geq 0$ such that if
$\mcS=(\{Y_i\}_{i=1}^m,\{X_j\}_{j=1}^n,|\psi\rangle)$ is an $\epsilon$-optimal strategy of
$\pm1$-valued observables, and $0\leq \epsilon \leq \frac{1}{4(m+n)}$, then
\begin{equation}\label{eq:norm} \left\|\left(\sum_{j=1}^n w_{ij}(\Id\otimes
X_j)-r_i(Y_i\otimes \Id)\right)|\psi\rangle\right\|= O(\epsilon^{1/4}), \end{equation} for
all $1\leq i\leq m$, and the constants hidden in the $O(\epsilon^{1/4})$ depends only on
the size of the question sets $m$ and $n$. \end{theorem}

In other words, \cref{thm:xor} establishes that near-optimal strategies are
state-dependent approximate representation of the XOR algebra (given in
\cref{def:xor_alg}). With this fact, we can establish the following result:

\begin{proposition}\label{prop:XOR_approx} Let
$\mcS=(\{Y_i\}_{i=1}^n,\{X_j\}_{j=1}^n,|\psi\rangle\in H_A\otimes H_B)$ be an
$\epsilon$-optimal strategy to an XOR game $\mcG$ where $|\psi\rangle$ has
reduced density matrix $\rho=\in \mcL(H_B)$, then the observables
$\{X_1,\ldots,X_n\}$ are an $(O(\epsilon^{1/4}),\rho)$-representation of the solution
algebra $\mcC(\mcG)$. Additionally, the
$(O(\epsilon^{1/4}),\rho)$-representation is $(O(\epsilon^{1/4}),\rho)$-tracial.
\end{proposition}

\begin{proof}
If $\mcS$ is an $\epsilon$-optimal strategy for the XOR game $\mcG$ and
$\rho=\lambda^*\lambda$ is the reduced density matrix of the state on $H_B$, since each
$X_j$ is a self-adjoint unitary, it only remains to show that the remaining relations in
\cref{def:xor_alg} hold approximately. In particular, we claim that
\begin{equation*}
\left\|r_i^2\Id-\left(\sum_{j=1}^n w_{ij}X_j\right)^2\right\|_\rho =
O(\epsilon^{1/4}),
\end{equation*} for all $1\leq i\leq m$.
This holds by \cref{thm:xor} and \cref{lem:frob_com}, as they show that in general
\begin{equation}\label{eqn:frob_norm}
\left\|\sum_{j=1}^n w_{ij}X_j\lambda-\lambda r_i\overline{Y_i}\right\|_F=
O(\epsilon^{1/4}).
\end{equation}
From this, it follows that
\begin{align*}
&\left\|r_i\Id-\left(\sum_{j=1}^n w_{ij}X_j\right)\right\|_f\\
=&\left\|r_i^2\lambda-\left(\sum_{j=1}^n w_{ij}X_j\right)^2\lambda\right\|_F\\ 
\leq & \left\|r_i^2\lambda-\sum_{j=1}^n
w_{ij}X_j r_i \lambda \overline{Y}_i \right\|_F+ \left\|\sum_{j=1}^n w_{ij}X_j
r_i \lambda \overline{Y}_i -\left(\sum_{j=1}^n
w_{ij}X_j\right)\left(\sum_{\ell=1}^n w_{i\ell}X_\ell\right)\lambda\right\|_F\\ \leq &
\left\|r_i^2\lambda\overline{Y}_i-\sum_{j=1}^n w_{ij}X_j r_i \lambda
\right\|_F+\sum_{j=1}^n |w_{ij}| \left\|X_j\left(r_i \lambda \overline{Y}_i
-\sum_{\ell=1}^n w_{i\ell}X_\ell\lambda\right)\right\|_F\\ \leq &
|r_i|\left\|r_i\lambda\overline{Y}_i-\sum_{j=1}^n w_{ij}X_j \lambda
\right\|_F+\sum_{j=1}^n |w_{ij}| \left\|r_i \lambda \overline{Y}_i -\sum_{\ell=1}^n
w_{i\ell}X_\ell\lambda\right\|_F
\end{align*}
is $O(\epsilon^{1/4})$ by equation \cref{eqn:frob_norm}. That this approximate
representation is $(O(\epsilon^{1/4}),\lambda)$-tracial follows easily from
\cref{lem:frob_com} and equation \cref{eqn:frob_norm}.
\end{proof}

\begin{corollary}\label{cor:XOR} If the state $|\psi\rangle$ in the
$\epsilon$-optimal strategy $\mcS$ for the XOR nonlocal game $\mcG$ is maximally
entangled then the observables $\{X_1,\ldots,X_n\}$ are a state-independent
$O(\epsilon^{1/4})$-representation of $\mcC(\mcG)$.
 \end{corollary}

To obtain the result in the general case, we apply our \cref{lem:rounding} to obtain the following result.

\begin{proposition}\label{prop:res3}
If $\mcS$ is an $\epsilon$-optimal strategy for an XOR nonlocal game $\mcG$, then
restricted to a subspace of $H_B$, Bob's measurement observables are a state-independent
$O(\epsilon^{1/8})$-representation of the XOR algebra $\mcC(\mcG)$.
\end{proposition}

\begin{proof}
For the XOR case, \cref{prop:XOR_approx} establishes that any $\epsilon$-perfect strategy
for an XOR game $\mcG$ results in an $O(\epsilon^{1/4})$-tracial
$(O(\epsilon^{1/4},\rho)$-representation of the XOR algebra $\mcC(\mcG)$. Again, because
the operators are $\pm 1$-valued observables, the approximate representation is exact for
the self-adjoint unitary relations. Hence, by \cref{lem:rounding} we obtain a
state-independent $O(\epsilon^{1/8})$-representation of the XOR algebra.
\end{proof}

Our last task is to determine the optimality of strategies arising from
$\epsilon$-representations in the $\|\cdot\|_f$-norm of the XOR algebra.

\begin{proposition}\label{prop:XOR_alg}
If $\phi$ is a bounded $\epsilon$-representations
of the solution algebra $\mcC(\mcG)$ on a Hilbert space $H$, then there is an
$O(\epsilon^2)$-optimal strategy for the corresponding XOR game using a maximally
entangled state.
\end{proposition}

\begin{proof}
Let $\phi$ be an $\epsilon$-representation of $\mcC(\mcG)$. Start by defining the
operators $\phi(x_j)$ for all $1\leq j\leq n$ on the Hilbert space $H_B$. We note that
theses may not be $\pm1$-valued observables, but they are close. In particular, by
\cref{lem:sa_un_stab}, we can find a nearby $\pm1$-valued observables $X_j$ for all $1\leq
j\leq n$ such that each $\|X_j-\phi(s_j)\|_f\leq 2\epsilon$. Let these self-adjoint
unitaries $\{X_1,\ldots,X_j\}$ be the observables in Bob's quantum strategy. To start
building Alice's strategy we first define
\begin{equation*}
Z_i=\frac{1}{r_i}\sum_j w_{ij}X_j^\top,
\end{equation*} for each $1\leq
i\leq m$. By the \cref{lem:replacement}, there exists a constant $K_0$ such that
$\|Z_i^2-\Id\|_f\leq K_0\epsilon$. Moreover, by noting that each $Z_i$ is self-adjoint by
construction, we have that $\|Z_iZ_i^*-\Id\|_f, \|Z_i^*Z_i-\Id\|_f$, and
$\|Z_i-Z_i^*\|_f$, are all at most $K_0\epsilon$ as well. Hence, applying
\cref{lem:sa_un_stab} again we obtain self-adjoint unitaries $Y_i$ such that $Y_i^2=\Id$,
$Y_i^*=Y_i$ and $\|Z_i^\top-Y_i\|_f\leq2K_0\epsilon$ for all $1\leq i \leq n$. Then, if
Alice's strategy consists of the operators $\{Y_1,\ldots, Y_m\}$ as defined for each
$1\leq i \leq m$, and they share a maximally entangled state $|\tau\rangle \in H_B\otimes
H_B$, we observe that
\begin{align*} |\beta_q(\mcG)-\beta(\mcS;\mcG)|&=\left|\sum_i r_i-
\sum_{ij}w_{ij}\langle \psi|Y_i\otimes X_j|\psi\rangle\right|\\
&\leq \sum_i r_i\left|1-\langle \psi|Y_i\otimes {Z_i}^\top|\psi\rangle\right|\\
&\leq \sum_i \frac{r_i}{2}\left|\tilde{\tr}(2\Id-2Y_i^\top Z_i)\right|\\
&\leq \sum_i \frac{r_i}{2}\|Y_i^\top-Z_i\|_f^2\\
&\leq 2n\max_i\{r_i\}K_0\epsilon^2
=O(\epsilon^2),
\end{align*} as desired.
\end{proof}

It remains to show that the measurement operators in the rounded strategy using the
maximally entangled state are not too far from the operators in the original
$\epsilon$-optimal strategy. Towards this point, in the proof of \cref{lem:rounding} we
see that on the support of the projection $P$, each unitary $\widetilde{X}_j$ is close to
the starting self-adjoint unitary $X_j$. In particular, this
distance depends on the initial approximate representation. In the XOR case, we see that
measurements obtained from \cref{lem:rounding} are at most $O(\epsilon^{1/8})$-away in the
$\|\cdot\|_f$-norm on the subspace $\widetilde{H}$. Lastly, in the proof of
\cref{prop:XOR_alg} it follows from stable replacement that whenever we obtain a strategy
for the corresponding nonlocal game from an $O(\epsilon')$-representation, the measurement
operators are never more than $O(\epsilon')$-away from the measurement operators in the
initial strategy.

\begin{corollary}\label{cor:xor}
For any $\epsilon$-optimal quantum strategy $\mcS$
for an XOR nonlocal game, there is an $O(\epsilon^{1/4})$-optimal quantum strategy
$\tilde{\mcS}$ using a maximally entangled state $|\tilde{\psi}\rangle$, such that each
measurement in $\tilde{\mcS}$ is at most $O(\epsilon^{1/8})$-away from the measurement in
$\mcS$ with respect to $\|\cdot\|_f$ on the local support of $|\tilde{\psi}\rangle$ on
$H_B$.
\end{corollary}

%%%%%%%%%%%%%%%%%%%%%%%%%%%%%%%%%%%%%%%%%%%%%%%%%%%

\section*{Aknowledgements}
The author would like to thank William Slofstra and Arthur Mehta for several helpful
discussions. They would also like to thank Yuming Zhao, Eric Culf, Taro Spirig, Denis
Rochette and the anonymous referees for valuable feedback on earlier drafts of this work.
This work was primarily completed while the author was at the University of Waterloo and
the Institute for Quantum Computing.

%%%%%%%%%%%%%%%%%%%%%%%%%%%%%%%%%%%%%%%%%%%%%%%%%

%\bibliography{bibfile} \bibliographystyle{alpha}

\begin{thebibliography}{B{\v{S}}CA18b}

\bibitem[B{\v{S}}CA18a]{BSCA18a}
Joseph Bowles, Ivan {\v{S}}upi{\'c}, Daniel Cavalcanti, and Antonio Ac{\'\i}n.
\newblock Device-independent entanglement certification of all entangled
  states.
\newblock {\em Physical review letters}, 121(18):180503, 2018.

\bibitem[B{\v{S}}CA18b]{BSCA18b}
Joseph Bowles, Ivan {\v{S}}upi{\'c}, Daniel Cavalcanti, and Antonio Ac{\'\i}n.
\newblock Self-testing of {P}auli observables for device-independent
  entanglement certification.
\newblock {\em Physical Review A}, 98(4):042336, 2018.

\bibitem[CHTW10]{CHTW10}
Richard Cleve, Peter Hoyer, Ben Toner, and John Watrous.
\newblock Consequences and limits of nonlocal strategies.
\newblock {\em arXiv:quant-ph/0404076}, Jan 2010.
\newblock arXiv: quant-ph/0404076.

\bibitem[CLS17]{CLS17}
Richard Cleve, Li~Liu, and William Slofstra.
\newblock Perfect commuting-operator strategies for linear system games.
\newblock {\em Journal of Mathematical Physics}, 58(1):012202, Jan 2017.
\newblock arXiv: 1606.02278.

\bibitem[CM14]{CM14}
Richard Cleve and Rajat Mittal.
\newblock Characterization of binary constraint system games.
\newblock In {\em Automata, Languages, and Programming: 41st International
  Colloquium, ICALP 2014, Copenhagen, Denmark, July 8-11, 2014, Proceedings,
  Part I 41}, pages 320--331. Springer, 2014.

\bibitem[Con76]{Con76}
Alain Connes.
\newblock Classification of injective factors cases $ii_1$, $ii_\infty$,
  $iii_{\lambda,\lambda\neq 1}$.
\newblock {\em Annals of Mathematics}, pages 73--115, 1976.

\bibitem[CSUU08]{CSUU08}
Richard Cleve, William Slofstra, Falk Unger, and Sarvagya Upadhyay.
\newblock Strong parallel repetition theorem for quantum xor proof systems.
\newblock {\em arXiv:quant-ph/0608146}, Apr 2008.
\newblock arXiv: quant-ph/0608146.

\bibitem[CVY23]{CVY23}
Michael Chapman, Thomas Vidick, and Henry Yuen.
\newblock Efficiently stable presentations from error-correcting codes.
\newblock {\em arXiv preprint arXiv:2311.04681}, 2023.

\bibitem[DP16]{DP16}
Kenneth~J Dykema and Vern Paulsen.
\newblock Synchronous correlation matrices and connes’ embedding conjecture.
\newblock {\em Journal of Mathematical Physics}, 57(1):015214, 2016.

\bibitem[Fri20]{Fri20}
Tobias Fritz.
\newblock Quantum logic is undecidable.
\newblock {\em Archive for Mathematical Logic}, Sep 2020.
\newblock arXiv: 1607.05870.

\bibitem[GH17]{GH17}
William~Timothy Gowers and Omid Hatami.
\newblock Inverse and stability theorems for approximate representations of
  finite groups.
\newblock {\em Sbornik: Mathematics}, 208(12), 2017.

\bibitem[Gol21]{Gol21}
Adina Goldberg.
\newblock Synchronous linear constraint system games.
\newblock {\em Journal of Mathematical Physics}, 62(3):032201, 2021.

\bibitem[Har24]{Har24}
Samuel~J Harris.
\newblock Approximate quantum 3-colorings of graphs and the quantum max 3-cut
  problem.
\newblock {\em arXiv preprint arXiv:2412.19405}, 2024.

\bibitem[HMPS19]{HMPS19}
J~William Helton, Kyle~P Meyer, Vern~I Paulsen, and Matthew Satriano.
\newblock Algebras, synchronous games, and chromatic numbers of graphs.
\newblock {\em New York J. Math}, 25:328--361, 2019.

\bibitem[HMV25]{HMV25}
Felix Huber, Victor Magron, and Jurij Vol{\v{c}}i{\v{c}}.
\newblock Positivity of state, trace, and moment polynomials and applications
  in quantum information.
\newblock In {\em Operator Theory}, pages 1--30. Springer, 2025.

\bibitem[Ji13]{Ji13}
Zhengfeng Ji.
\newblock Binary constraint system games and locally commutative reductions.
\newblock {\em arXiv preprint arXiv:1310.3794}, 2013.

\bibitem[JNV{\etalchar{+}}22]{JNVWY22}
Zhengfeng Ji, Anand Natarajan, Thomas Vidick, John Wright, and Henry Yuen.
\newblock Mip*= re.
\newblock {\em arXiv preprint arXiv:2001.04383v3}, 2022.

\bibitem[Kan20]{Kan19}
Jedrzej Kaniewski.
\newblock Weak form of self-testing.
\newblock {\em Physical Review Research}, 2, 2020.

\bibitem[KPS18]{KPS18}
Se-Jin Kim, Vern Paulsen, and Christopher Schafhauser.
\newblock A synchronous game for binary constraint systems.
\newblock {\em Journal of Mathematical Physics}, 59(3):032201, 2018.

\bibitem[LMP{\etalchar{+}}20]{Lup20}
Martino Lupini, Laura Man{\v{c}}inska, Vern~I Paulsen, David~E Roberson,
  Giannicola Scarpa, Simone Severini, Ivan~G Todorov, and Andreas Winter.
\newblock Perfect strategies for non-local games.
\newblock {\em Mathematical Physics, Analysis and Geometry}, 23(1):7, 2020.

\bibitem[MSZ23]{MSZ23}
Arthur Mehta, William Slofstra, and Yuming Zhao.
\newblock Positivity is undecidable in tensor products of free algebras.
\newblock {\em arXiv preprint arXiv:2312.05617}, 2023.

\bibitem[MY04]{MY03}
Dominic Mayers and Andrew Yao.
\newblock Self testing quantum apparatus.
\newblock {\em Quantum Info. Comput.}, 4(4):273–286, jul 2004.

\bibitem[Oza13]{Oza13}
Narutaka Ozawa.
\newblock About the connes embedding conjecture: algebraic approaches.
\newblock {\em Japanese Journal of Mathematics}, 8(1):147--183, 2013.

\bibitem[Pad23]{Pad23}
Connor Paddock.
\newblock {\em Near-optimal quantum strategies for nonlocal games, approximate
  representations, and BCS algebras}.
\newblock PhD thesis, University of Waterloo, 2023.

\bibitem[PSS{\etalchar{+}}16]{PSSTW16}
Vern~I. Paulsen, Simone Severini, Daniel Stahlke, Ivan~G. Todorov, and Andreas
  Winter.
\newblock Estimating quantum chromatic numbers.
\newblock {\em Journal of Functional Analysis}, 270(6):2188–2222, Mar 2016.

\bibitem[PT15]{PT15}
Vern~I Paulsen and Ivan~G Todorov.
\newblock Quantum chromatic numbers via operator systems.
\newblock {\em The Quarterly Journal of Mathematics}, 66(2):677--692, 2015.

\bibitem[Slo11]{Slof11}
William Slofstra.
\newblock Lower bounds on the entanglement needed to play {XOR} non-local
  games.
\newblock {\em Journal of Mathematical Physics}, 52(10):102202, 2011.

\bibitem[Slo18]{Slof18}
William Slofstra.
\newblock A group with at least subexponential hyperlinear profile.
\newblock {\em arXiv preprint arXiv:1806.05267}, 2018.

\bibitem[Slo19]{Slof19b}
William Slofstra.
\newblock The set of quantum correlations is not closed.
\newblock {\em Forum of Mathematics, Pi}, 7:e1, 2019.

\bibitem[SV18]{SV18}
William Slofstra and Thomas Vidick.
\newblock Entanglement in non-local games and the hyperlinear profile of
  groups.
\newblock {\em Annales Henri Poincar{\'e}}, 19(10):2979--3005, 2018.

\bibitem[SVW16]{SVW16}
Jamie Sikora, Antonios Varvitsiotis, and Zhaohui Wei.
\newblock Minimum dimension of a {H}ilbert space needed to generate a quantum
  correlation.
\newblock {\em Physical Review Letters}, 117(6):060401, 2016.

\bibitem[Tho18]{Thom18}
Andreas Thom.
\newblock Finitary approximations of groups and their applications.
\newblock In {\em Proceedings of the International Congress of Mathematicians
  (ICM 2018) (In 4 Volumes) Proceedings of the International Congress of
  Mathematicians 2018}, pages 1779--1799. World Scientific, 2018.

\bibitem[Tsi85]{Tsi85}
BS~Tsirelson.
\newblock Quantum analogues of bell’s inequalities. the case of two spatially
  divided domains.
\newblock {\em Zap. Nauchn. Sem. Leningrad. Otdel. Mat. Inst. Steklov.(LOMI)},
  142:174--194, 1985.
\newblock In Russian.

\bibitem[Tsi87]{Tsi87}
BS~Tsirelson.
\newblock Quantum analogues of bell’s inequalities. the case of two spatially
  divided domains.
\newblock {\em Journal of Soviet Mathematics}, 36(4):557--570, 1987.
\newblock Translated from Russian.

\bibitem[Vid22]{Vid22}
Thomas Vidick.
\newblock Almost synchronous quantum correlations.
\newblock {\em Journal of mathematical physics}, 63(2):022201, 2022.

\bibitem[WBMS16]{Wu16}
Xingyao Wu, Jean-Daniel Bancal, Matthew McKague, and Valerio Scarani.
\newblock Device-independent parallel self-testing of two singlets.
\newblock {\em Physical Review A}, 93(6), 2016.

\bibitem[Weh06]{Whe06}
Stephanie Wehner.
\newblock Tsirelson bounds for generalized {C}lauser-{H}orne-{S}himony-{H}olt
  inequalities.
\newblock {\em Physical Review A}, 73(2):022110, 2006.

\end{thebibliography}

\newcommand{\etalchar}[1]{$^{#1}$}

%%%%%%%%%%%%%%%%%%%%%%%%%%%%%%%%%%%%%%%%%%%%%%%%%

\appendix

\section{On the unitary parts of a matrices with a certain restriction property}\label{sec:parts}

\begin{proof}[Proof of \cref{lem:unitary_res}]\label{pf:un_res}
To begin let $Y=PXP$ and for now suppose that $P$ has rank $k$ and that $P$ is diagonal in
the standard basis i.e.~$P=\Id_{k}\oplus \mathbf{0}_{d-k}$. Since $Y$ is supported on
$Im(P)$ we can write $Y=\begin{pmatrix} \tilde{Y} & \mathbf{0}_{k\times d-k} \\
\mathbf{0}_{d-k,k} & \mathbf{0}_{d-k,d-k} \end{pmatrix}$ for some self-adjoint matrix
$\tilde{Y} \in M_k(\C)$. Now take $\tilde{U}$ to be a unitary part of $\tilde{Y}$ on
$\C^k$. To extend $\tilde{U}$ to matrix on $\C^d$ we let $U=\tilde{U}\oplus \Id_{d-k}$. In
addition to $U$ being a unitary on $\C^d$, we observe that $U$ restricts to a unitary on
$Im(P)$. Since $\tilde{Y}=\tilde{U}|\tilde{Y}|$, all that remains is to verify is if $|Y|$
equals $\begin{pmatrix} |\tilde{Y}| & \mathbf{0}_{k\times d-k} \\ \mathbf{0}_{d-k,k} &
\mathbf{0}_{d-k,d-k} \end{pmatrix}$. With this in mind, let $\tilde{V}$ be the unitary in
$M_k(\C)$ that diagonalizes $\tilde{Y}$ (i.e.~$\tilde{Y}=\tilde{V}\tilde{D}\tilde{V}^*$
for a diagonal matrix $\tilde{D}\in M_k(\C)$). We observe that
\begin{align*} |Y|&=\bigg| \begin{pmatrix} \tilde{Y} & \mathbf{0}_{k\times d-k} \\
\mathbf{0}_{d-k,k} & \mathbf{0}_{d-k,d-k} \end{pmatrix}\bigg|\\ & =\bigg| \begin{pmatrix}
\tilde{V} & \mathbf{0}_{k\times d-k} \\ \mathbf{0}_{d-k,k} & \Id_{d-k}
\end{pmatrix}\begin{pmatrix} \tilde{D} & \mathbf{0}_{k\times d-k} \\ \mathbf{0}_{d-k,k} &
\mathbf{0}_{d-k,d-k} \end{pmatrix}\begin{pmatrix} \tilde{V}^* & \mathbf{0}_{k\times d-k}
\\ \mathbf{0}_{d-k,k} & \Id_{d-k} \end{pmatrix}\bigg| \\ &=\begin{pmatrix} \tilde{V} &
\mathbf{0}_{k\times d-k} \\ \mathbf{0}_{d-k,k} & \Id_{d-k}
\end{pmatrix}\bigg|\begin{pmatrix} \tilde{D} & \mathbf{0}_{k\times d-k} \\
\mathbf{0}_{d-k,k} & \mathbf{0}_{d-k,d-k} \end{pmatrix}\bigg|\begin{pmatrix} \tilde{V}^* &
\mathbf{0}_{k\times d-k} \\ \mathbf{0}_{d-k,k} & \Id_{d-k} \end{pmatrix}\\
&=\begin{pmatrix} \tilde{V} & \mathbf{0}_{k\times d-k} \\ \mathbf{0}_{d-k,k} & \Id_{d-k}
\end{pmatrix}\begin{pmatrix} |\tilde{D}| & \mathbf{0}_{k\times d-k} \\ \mathbf{0}_{d-k,k}
& \mathbf{0}_{d-k,d-k} \end{pmatrix}\begin{pmatrix} \tilde{V}^* & \mathbf{0}_{k\times d-k}
\\ \mathbf{0}_{d-k,k} & \Id_{d-k} \end{pmatrix}\\ &=\begin{pmatrix} |\tilde{Y}| &
\mathbf{0}_{k\times d-k} \\ \mathbf{0}_{d-k,k} & \mathbf{0}_{d-k,d-k} \end{pmatrix},
\end{align*} as required. Next, we verify that
\begin{equation*}
U|Y|=\begin{pmatrix}\tilde{U} & \mathbf{0}_{k\times d-k} \\ \mathbf{0}_{d-k,k} & \Id_{d-k}
\end{pmatrix}\begin{pmatrix}|\tilde{Y}| & \mathbf{0}_{k\times d-k} \\ \mathbf{0}_{d-k,k} &
\mathbf{0}_{d-k,d-k}\end{pmatrix} =\begin{pmatrix} \tilde{U}|\tilde{Y}| &
\mathbf{0}_{k\times d-k} \\ \mathbf{0}_{d-k,k} & \mathbf{0}_{d-k,d-k}\end{pmatrix}=Y.
\end{equation*}

In the case that $P$ is not diagonal in the standard basis there exists a unitary $W\in
M_d(\C)$ so that $WPW^*$ is diagonal. In this case, we repeat the steps above with
$WYW^*=\begin{pmatrix} \tilde{Y} & \mathbf{0}_{k\times d-k} \\ \mathbf{0}_{d-k,k} &
\mathbf{0}_{d-k,d-k} \end{pmatrix}$. Doing so provides us with a unitary $U$ satisfying
$U|WYW^*|=WYW^*$ which is equivalent to $W^*UW|Y|=Y$. In this case, we take the unitary
part of $Y$ to be $U'=W^*UW$. All that remains is to show that $U'$ restricts to a unitary
on $Im(P)$. Recall that the non-zero eigenvectors $\{|v_1\rangle,\ldots,|v_k\rangle\}$ of
$P$ are an orthonormal basis for $Im(P)$. Then, for any $1\leq i \leq k$ we have

\begin{equation*}
U'|v_i\rangle=U'W^*|e_i\rangle=W^*U|e_i\rangle=\sum_{j=1}^k
\gamma_jW|e_j\rangle=\sum_{j=1}^k\gamma_j|v_j\rangle\in Im(P),
\end{equation*}
where $\sum_{i=1}^k|\gamma_j|^2=1$ for all $1\leq i \leq k$. Moreover, the columns of $U'$
restricted to this basis are orthonormal, completing the proof.
\end{proof}

\end{document}